\documentclass[journal]{IEEEtran}
\IEEEoverridecommandlockouts
\usepackage{cite}
\usepackage{amsmath,amssymb,amsfonts}
\usepackage{bm}
\usepackage{algorithmic,algorithm}
\usepackage{graphicx}
\usepackage{textcomp}
\usepackage{xcolor}
\usepackage{booktabs}
\usepackage{subcaption}
\usepackage{multirow}
\usepackage{makecell}
\usepackage{subfig}
\usepackage{etoolbox}

\usepackage{enumitem}
\newtheorem{assumption}{Assumption}
\newtheorem{proposition}{Proposition}

\def\BibTeX{{\rm B\kern-.05em{\sc i\kern-.025em b}\kern-.08em
T\kern-.1667em\lower.7ex\hbox{E}\kern-.125emX}}
\begin{document}

\title{Domain-Adaptive Communication-Rate Optimization for Sim-to-Real Humanoid-Robot Wireless XR Teleoperation
\author{Caolu Xu, Zhiyong Chen, \emph{Senior Member, IEEE}, Meixia Tao, \emph{Fellow, IEEE}, Li Song, \emph{Senior Member, IEEE}, \\ Feng Yang, and Wenjun Zhang, \emph{Fellow, IEEE}}

\thanks{The authors are with the Cooperative Medianet Innovation Center, the School of Information Science and Electronic Engineering, Shanghai Jiao Tong University, Shanghai 200240, China (e-mail: \{021034910015, zhiyongchen, mxtao, song\_li, yangfeng, zhangwenjun\}@sjtu.edu.cn). }}
\maketitle

\begin{abstract}
Wireless extended reality (XR) teleoperation provides embodied interaction capability for collecting humanoid robot demonstrations, but the large-scale adoption is restricted by the overhead of high-frequency motion transmission. This paper develops a system framework that integrates sampling, transmission, interpolation, and reconstruction and formulates a communication-rate optimization that aims to minimize the communication energy while maintaining the reconstruction accuracy of robot motion trajectories through dimension-wise sampling-rate control. Since acquiring real-time feedback from physical robots is limited by hardware costs, it is necessary to solve the problem through simulator interaction with offline real-domain data correction. To guide sim-to-real adaptation, we provide a PAC-Bayes generalization characterization that reveals the effects of latent density-ratio estimation, finite-sample deviation, and encoder bias. Building on this analysis, we propose a proximal policy optimization (PPO) method with density-ratio weighting and trust-region regularization. Experiments on public humanoid teleoperation dataset show that the proposed method improves the tradeoff between reconstruction error and communication energy consumption under sim-to-real distribution shift. We further analyze the effectiveness of the proposed algorithm across various wireless channels and dynamic motion trajectories.
\end{abstract}

\section{Introduction}
Embodied intelligence has recently emerged as a highly promising direction in both industry and academia. Representative paradigms, such as vision-language-action (VLA)~\cite{vla} and world models (WMs)~\cite{wm}, rely on large-scale demonstration data to train generalizable embodied policies. In this context, extended reality (XR) enables wireless teleoperation through immersive interaction and intuitive human interfaces \cite{meta_quest3, apple_vision_pro}, making it particularly suitable for scalable collection of high-quality human demonstrations.

Despite these advantages, wireless XR teleoperation still faces several critical challenges. From the communication perspective, accurate motion reconstruction requires high-frequency motion sampling and transmission. Such frequent updates impose  non-negligible energy consumption on the battery-powered XR device. In addition, wireless channels are inherently time-varying due to user mobility, blockage, and interference, which further complicates reliable low-latency transmission. Different dimensions contribute unequally to robot-side reconstruction fidelity, implying that uniform full-rate transmission of all components can be energy-inefficient.

From the learning perspective, the evaluation metric is the execution performance of physical robots in real-world environments. 
However, large-scale real-world interaction is often impractical for training communication-aware teleoperation policies due to safety constraints, high hardware costs, and limited data collection efficiency. As a result, policy training is typically conducted in simulation environments, where abundant rollouts can be generated efficiently. Nevertheless, the discrepancy between simulated and real-world environments introduces a sim-to-real distribution shift, which may degrade policy performance during real deployment.

These challenges motivate the need to design a wireless XR teleoperation framework that jointly balances device resource consumption and motion reconstruction accuracy, while explicitly accounting for the sim-to-real gap.
\begin{figure*}[t]
\centering{\includegraphics[width=0.8\textwidth]{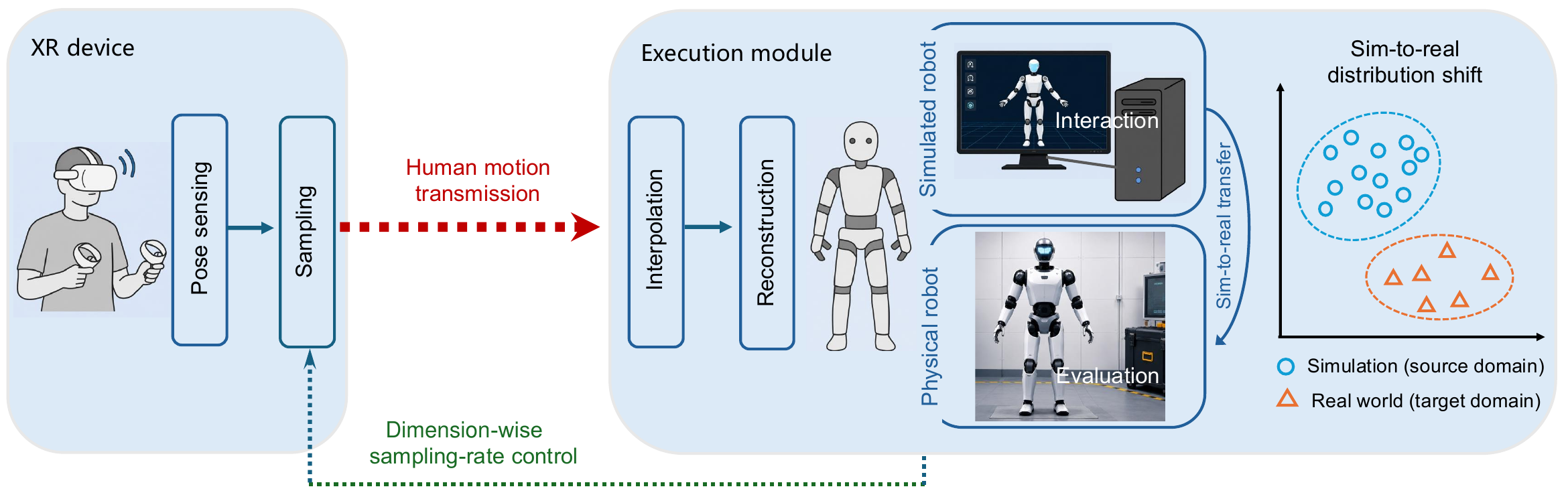}}
\caption{Overview of the proposed sim-to-real wireless XR teleoperation framework for humanoid robot.}
\label{system}
\end{figure*}

\subsection{Related Work}
Recently, several works have investigated wireless XR systems. For XR video streaming, latency optimization of prediction, communication, and computing was studied in~\cite{time_prediction}, while cross-frame resource allocation with proactive caching for 360$^\circ$ video was investigated in~\cite{proactive_caching}. For interactive XR, wireless multiplayer systems with edge computing were modeled in~\cite{VRgame_Zhu}, and an edge-device collaborative framework leveraging foreground-background separation was developed in~\cite{my_interactive_VR}. Beyond XR services, synchronization between physical devices and digital models in the metaverse has been studied to enable timely and accurate model updates~\cite{JSAC_task_oriented_meta,JSAC_sample}. These studies demonstrate the importance of cross-layer resource management for wireless XR. However, most existing works focus on media transmission or system-level coordination, and few have investigated wireless transmission for robot motion reconstruction while explicitly considering the sim-to-real gap in XR teleoperation.

In humanoid robot learning, sim-to-real transfer has become an important topic.
Reinforcement learning (RL) algorithms have been widely used to train humanoid
control policies for whole-body skills, motion imitation, and recovery behaviors.
HoST~\cite{huang2025host} learns standing-up control for humanoids across diverse
postures, demonstrating the potential of RL-based policies for real-world humanoid
recovery tasks. ASAP~\cite{he2025asap} uses real-world data to mitigate dynamics mismatch and improve the transferability of humanoid skills learned in
simulation, while Humanoid Policy~\cite{qiu2025humanoid} studies
cross-embodiment learning from human and humanoid demonstrations. More
recently, DoorMan~\cite{ cvpr2026} investigates
humanoid pixel-to-action policy transfer and shows that vision-based
loco-manipulation policies trained in simulation can be transferred to real
humanoids for articulated-object interaction. 
These methods primarily aim to enhance the transferability of humanoid manipulation and locomotion skills. 
In contrast, our work focuses on communication-rate policy learning in wireless XR teleoperation. Specifically, the policy is trained mainly using simulator rollouts, while offline real-world teleoperation
trajectories are used to characterize deployment bias and mitigate the sim-to-real distribution shift.

Domain adaptation has shown promise in addressing distribution shifts in communication systems, such as mmWave-based human activity recognition~\cite{zhao2025federated}, collaborative learning on mobile edge devices~\cite{zhou2024mecda}, and cross-domain wireless localization under non-line-of-sight conditions~\cite{chen2025crossdomain}. In general, domain adaptation aims to transfer knowledge from a source domain, where effective model training is feasible, to a target domain with a different distribution, by reducing source-target discrepancy. Representative
techniques include instance reweighting and density-ratio estimation
~\cite{KLIEP,uLSIF}, as well as discrepancy-based feature alignment approaches such as maximum mean discrepancy (MMD) and correlation alignment
~\cite{mmd,long2015learning,sun2016deep}. In this paper, domain adaptation provides a principled way to correct simulator-induced bias using offline real-domain trajectories, preserving the effectiveness of communication rate decisions without requiring online interaction with the physical robot during policy training. Compared with existing methods developed for model adaptation under relatively fixed input-output mappings, wireless XR teleoperation involves sequential communication control, where current decisions affect subsequent system states and performance. This motivates domain-adaptive policy learning rather than direct application of conventional domain adaptation methods. 

\subsection{Contributions}
In this paper, we investigate sim-to-real communication-rate optimization in wireless XR teleoperation for humanoid robot. The main contributions are summarized as follows.

\begin{itemize}
\item We propose a sim-to-real wireless XR teleoperation framework that considers dimension-wise motion sampling, energy-aware wireless communication, robot reconstruction, and platform execution. Based on this framework, we formulate a long-term communication rate optimization problem that balances the motion-reconstruction accuracy and device communication energy under sim-to-real distribution shift.
\item We develop a PAC-Bayes \cite{pac} generalization characterization  with latent density-ratio estimation and encoder-induced representation bias. This characterization decomposes the real-domain performance gap into the importance-weighted empirical risk, density-ratio approximation error, finite-sample deviation, and encoder representation bias. This analysis provides theoretical guidance for designing the representation learning and policy optimization procedure.

\item We introduce latent density-ratio correction into proximal policy optimization (PPO) \cite{ppo} for communication-rate optimization and develop a domain-adaptive training framework. First, an encoder is initialized through a MMD-based warm-up on simulator samples and offline real-domain trajectories. The initialized encoder is then frozen to provide a stable latent space for density-ratio estimation and importance-weighted PPO. Finally, with the policy network fixed, the encoder is fine-tuned using trust-region regularization \cite{trust} and weighted MMD, improving latent alignment while preventing uncontrolled representation drift. The experimental results demonstrate that the proposed algorithm achieves a better reconstruction-energy tradeoff than baseline methods under sim-to-real distribution shift.
\end{itemize}

The rest of this paper is organized as follows. Section~II presents the system model and formulates the sim-to-real communication-rate optimization problem. Section~III provides the PAC-Bayes generalization analysis with latent density ratio and representation bias. Section~IV develops the domain-adaptive policy optimization algorithm. Section~V presents simulation results. Section~VI concludes the paper.

\section{System Model and Problem Formulation}
\label{sec:system_model}

\textcolor{black}{In this section, we present a framework for the sim-to-real wireless XR teleoperation system. As illustrated in Fig.~\ref{system}, the XR device first captures human motion trajectories through pose sensing, followed by dimension-wise sampling for communication-efficient transmission. The sampled human motion data are transmitted to the execution module, where interpolation and reconstruction are performed to recover robot motions. 
The system aims to achieve accurate, real-time motion reconstruction on the physical robot while reducing unnecessary device communication energy under time-varying wireless conditions and the sim-to-real distribution shift.}

\subsection{Framework and Cross-System Architecture}
\begin{figure*}[t]
\centering{\includegraphics[width=0.8\textwidth]{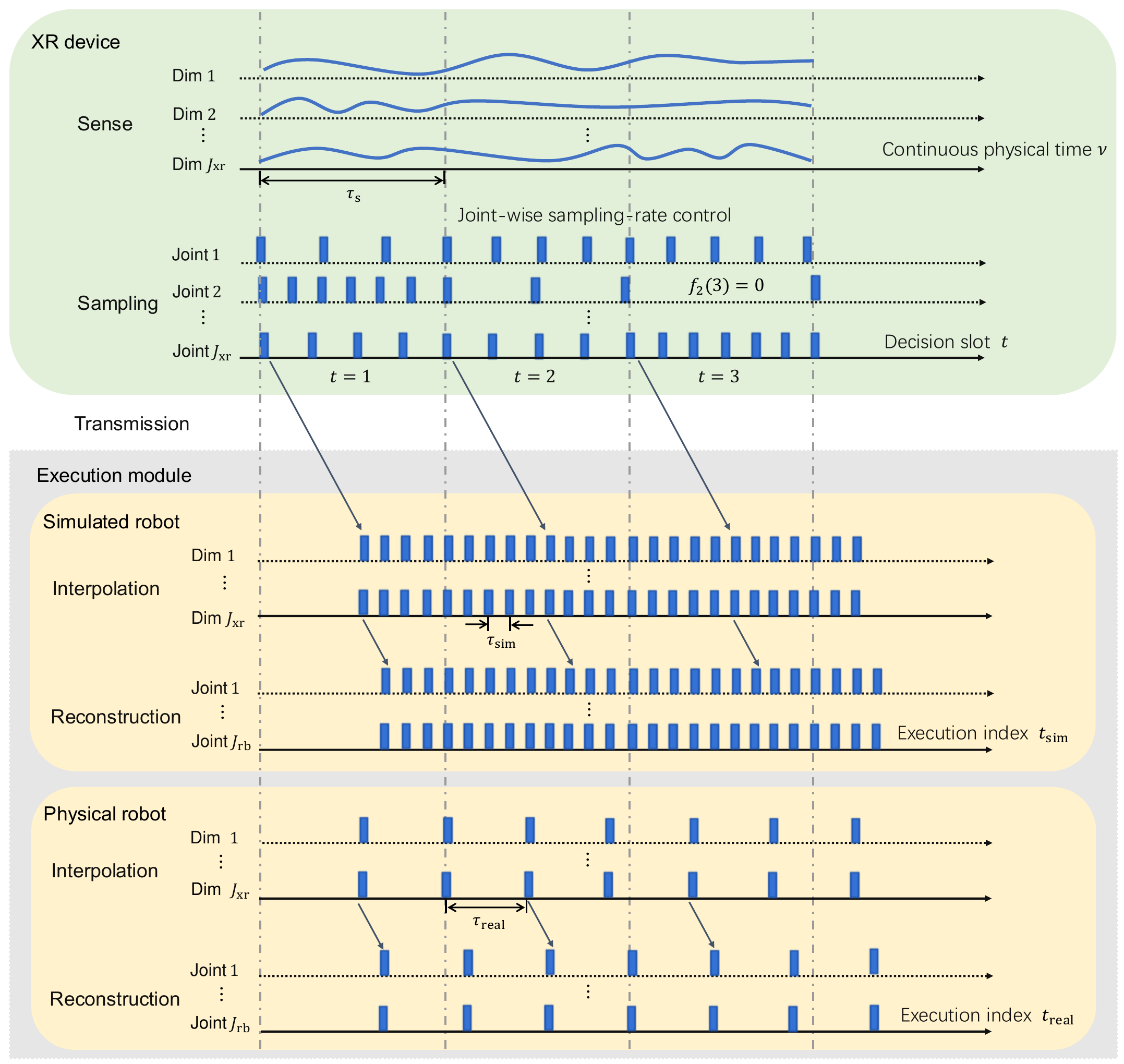}}
	\caption{Timeline of the proposed wireless XR teleoperation framework.}
	\label{timeline}
\end{figure*} 

The timeline of  the proposed framework is shown in Fig. \ref{timeline}. The teleoperation process is described on three time scales: continuous physical time $\nu$ for the raw XR motion, decision slots $t$ for sampling-rate control, and execution indices $t_\dagger$ for simulation or robot updates, where $\dagger\in\{\mathrm{sim},\mathrm{real}\}$ denotes the execution platform. 

The decision process is modeled in discrete slots indexed by $t$, where each slot has duration $\tau_s$ and corresponds to interval
\begin{equation}
\label{eq:decision_interval}
\mathcal{T}_t=[(t-1)\tau_s,\ t\tau_s).
\end{equation}
At the beginning of slot $t$, sampling rates are selected and remain effective over $\mathcal{T}_t$.

\paragraph*{XR and robot configuration spaces}
$\nu\in\mathbb{R}_{+}$ denotes continuous physical time. The human motion trajectory captured by the XR device is denoted by
\begin{equation}
\label{eq:vr_continuous}
\mathbf{q}^{\mathrm{xr}}(\nu)
=
\bigl[
q_1^{\mathrm{xr}}(\nu),\ldots,
q_{J_{\mathrm{xr}}}^{\mathrm{xr}}(\nu)
\bigr]^\top
\in\mathbb{R}^{J_{\mathrm{xr}}},
\ \nu\in\mathbb{R}_{+},
\end{equation}
where $J_{\mathrm{xr}}$ is the number of tracked XR sampling dimension.

The humanoid robot configuration space is represented by
\begin{equation}
\label{eq:robot_space}
\mathbf{q}^{\mathrm{rb}}\in\mathbb{R}^{J_{\mathrm{rb}}},
\end{equation}
where $J_{\mathrm{rb}}$ is the number of controllable robot joints.
In general, $J_{\mathrm{xr}}\neq J_{\mathrm{rb}}$, and the two spaces have different kinematic structures.

\paragraph*{Execution time scales}
The execution platform is indicated by $\dagger\in\{\mathrm{sim},\mathrm{real}\}$,
with $\mathrm{sim}$ and $\mathrm{real}$ corresponding to the simulated and physical robot, respectively. 
 For platform $\dagger$, the execution frequency is $f_{\dagger}$, the execution period is $\tau_{\dagger}=1/{f_{\dagger}}$, and the execution index is $t_{\dagger}\in\mathbb{Z}_{+}$.

To satisfy the update requirements of platform $\dagger$, the XR motion trajectory is aligned to the execution grid.
For the ideal reference trajectory, each sampling dimension is assumed to be transmitted over a lossless channel at its highest admissible rate.
Let $f_j^{\max}$ denote the highest admissible rate of the $j$-th sampling dimension, and let $T_j^{\max}=1/f_j^{\max}$ be the corresponding sampling period.
The reference sampling-time set $\mathcal{V}_j^{\max}$ consists of equally spaced physical instants with interval $T_j^{\max}$.
For each dimension $j$, the ideal interpolated XR trajectory is defined as
\begin{equation}
\label{eq:ideal_vr_aligned_joint}
q_j^{\mathrm{xr}}(t_{\dagger})
=
\mathcal{I}_{j,\dagger}^{\mathrm{ref}}
\!\left(
t_{\dagger}\tau_{\dagger};
\left\{
\bigl(\nu,q_j^{\mathrm{xr}}(\nu)\bigr):
\nu\in\mathcal{V}_j^{\max}
\right\}
\right),
\end{equation}
where $\mathcal{I}_{j,\dagger}^{\mathrm{ref}}(\cdot)$ is the reference interpolation operator that evaluates the $j$-th XR sampling dimension at execution instant $t_{\dagger}\tau_{\dagger}$ from the reference samples. 
Stacking all dimensions gives $\mathbf{q}^{\mathrm{xr}}(t_{\dagger})
=
\bigl[
q_1^{\mathrm{xr}}(t_{\dagger}),\ldots,
q_{J_{\mathrm{xr}}}^{\mathrm{xr}}(t_{\dagger})
\bigr]^\top$.

\paragraph*{Retargeting operator}
The mapping from the XR sample space to the robot joint space is denoted by
\begin{equation}
\label{eq:retarget_operator}
    \mathcal{R}:\mathbb{R}^{J_{\mathrm{xr}}}\rightarrow\mathbb{R}^{J_{\mathrm{rb}}}.
\end{equation}
Then the ideal robot joint trajectory on platform $\dagger$ is 
\begin{equation}
\label{eq:ideal_robot}
    \mathbf{q}^{\mathrm{rb}\star}(t_{\dagger})
    =\mathcal{R}\big(\mathbf{q}^{\mathrm{xr}}(t_{\dagger})\big), \ \dagger\in\{\mathrm{sim},\mathrm{real}\}.
\end{equation}
The specific realization of $\mathcal{R}$, such as inverse kinematics, does not affect the subsequent analysis.

\subsection{Sampling, Transmission, Interpolation, and Reconstruction}
\label{subsec:sampling}
\paragraph*{Dimension-wise sampling rates}
To account for different motion characteristics and transmission requirements across sampling dimensions, we allow the admissible sampling-rate set to depend on the dimension index. At the beginning of the slot $t$, the sampling rate assigned to the $j$-th sampling dimension is denoted by $f_j(t)$, where
\begin{equation}\label{fj_costrant}
f_j(t)\in\mathcal{F}_j=\{0,f_j^{(1)},\ldots,f_j^{(F_j)}\},\  j=1,\ldots,J_{\mathrm{xr}}.
\end{equation}

Here, $\mathcal{F}_j$ is the admissible sampling-rate set of the $j$-th XR sampling dimension, and $F_j$ is the number of nonzero rate levels. The nonzero rate levels are ordered as $0<f_j^{(1)}<\cdots<f_j^{(F_j)}=f_j^{\max}.$ In particular, $f_j(t)=0$ indicates that there is no transmission of dimension $j$ in slot $t$.
The dimension-wise sampling rate in slot $t$ is
\begin{equation}
\mathbf{f}(t)=\bigl[f_1(t),\ldots,f_{J_{\mathrm{xr}}}(t)\bigr]^\top.
\end{equation}

Given $f_j(t)$, let $\mathcal{V}_j(t)\subseteq\mathcal{T}_t$ denote the set of physical time stamps at which samples of the $j$-th XR dimension are received within slot $t$. If $f_j(t)=0$, no sample of dimension $j$ is transmitted and $\mathcal{V}_j(t)=\varnothing.$
If $f_j(t)>0$, $\mathcal{V}_j(t)$ is generated by the sampling clock associated with rate $f_j(t)$ over $\mathcal{T}_t$. Thus, $t$ indexes the decision slot, whereas $\nu\in\mathcal{V}_j(t)$ denotes a physical time stamp used for interpolation of the $j$-th XR sampling dimension.

\paragraph*{Wireless transmission and energy cost}
We consider a orthogonal frequency division multiplexing (OFDM) wireless link with total bandwidth $B$ and use an equivalent-link rate model after subcarrier aggregation. Let $P(t)$ denote the total transmit power of the XR device in slot $t$. The equivalent channel gain $g(t)$ captures the time-varying wireless conditions. We assume that $g(t)$ is available at the beginning of slot $t$. $N_0$ denotes the spectral density of noise power. The transmission rate $C(t)$ in slot $t$ is given by
\begin{equation}
\label{eq:shannon_rate}
C(t)
=
B\log_2\!\left(
1+\frac{P(t)g(t)}{N_0B}
\right).
\end{equation}

Let $\mathbf{d}=[d_1,\ldots,d_{J_{\mathrm{xr}}}]^\top$ denote the payload-size vector for one update of the XR samples, where $d_j$ is the number of bits transmitted for one update of dimension $j$.
Given the sampling-rate vector $\mathbf{f}(t)$ in slot $t$, the transmission load over the slot is
\begin{equation}
\label{eq:payload}
D(t)
=
\tau_s \mathbf{d}^{\top} \mathbf{f}(t).
\end{equation}

The transmission constraint in slot $t$ is given by
\begin{equation}
\label{eq:payload_constraint}
D(t)\le \tau_s C(t).
\end{equation}

Effective communication energy is not limited to the radio frequency (RF)  transmit power. In addition to RF transmit power, packetization,  protocol-related, and other communication pipeline overheads also contribute to the energy consumption of XR device. For tractability, these additional overheads are incorporated through an effective energy coefficient $\zeta\geq 1$. The effective communication energy in slot $t$ is defined as
\begin{equation}
\label{eq:energy_raw}
E(t)=\zeta \tau_s P(t).
\end{equation}

\paragraph*{Sample interpolation}
For the $j$-th XR sampling dimension, the sample received at the execution module at physical instant $\nu\in\mathcal{V}_j(t)$ is denoted by $\tilde q_j^{\mathrm{xr}}(\nu)$.
To match the execution time scale of platform $\dagger$, the received XR samples are interpolated on the execution grid.
For each dimension $j$, the reconstructed XR sample value at execution index $t_{\dagger}$ is obtained from the samples received no later than the execution instant $t_{\dagger}\tau_{\dagger}$
\begin{equation}
\label{eq:interp_joint}
\begin{aligned}
\hat q_j^{\mathrm{xr}}(t_{\dagger})
=
\mathcal{I}_{j,\dagger}
\!\Bigg(
t_{\dagger}\tau_{\dagger};\bigg\{
\bigl(\nu,\tilde q_j^{\mathrm{xr}}(\nu)\bigr):
\nu\in\bigcup_{t\ge 1}\mathcal{V}_j(t),\
\nu\le t_{\dagger}\tau_{\dagger}\bigg\}
\Bigg),
\end{aligned}
\end{equation}
where $\mathcal{I}_{j,\dagger}(\cdot)$ denotes the interpolation operator for dimension $j$ on platform $\dagger$. 
If no prior sample is available, an initial pose is used.
Stacking all reconstructed dimensions gives the interpolated XR trajectory $\hat{\mathbf{q}}^{\mathrm{xr}}(t_{\dagger})
=
\bigl[
\hat q_1^{\mathrm{xr}}(t_{\dagger}),\ldots,
\hat q_{J_{\mathrm{xr}}}^{\mathrm{xr}}(t_{\dagger})
\bigr]^\top$.
 
\paragraph*{Reconstruction error}
The reconstructed robot joint trajectory is obtained by applying the retargeting operator to the interpolated XR motion
\begin{equation}
\label{eq:robot_recon}
\hat{\mathbf{q}}^{\mathrm{rb}}(t_{\dagger})
=
\mathcal{R}\bigl(\hat{\mathbf{q}}^{\mathrm{xr}}(t_{\dagger})\bigr).
\end{equation}

Since the reconstruction fidelity is evaluated on the execution time scale, the joint reconstruction error is defined as
\begin{equation}
\label{eq:joint_error}
e_{\mathrm{joint}}(t_{\dagger})
=
\left(
\frac{
\left\|
\mathbf{W}^{1/2}
\bigl(
\hat{\mathbf{q}}^{\mathrm{rb}}(t_{\dagger})
-
\mathbf{q}^{\mathrm{rb}\star}(t_{\dagger})
\bigr)
\right\|_2^2
}{
\sum_{j=1}^{J_{\mathrm{rb}}} w_j
}
\right)^{1/2},
\end{equation}
where $\mathbf{W}=\mathrm{diag}(w_1,\ldots,w_{J_{\mathrm{rb}}})$ with $w_j\ge 0$ is the joint-importance weighting matrix, and $\mathbf{q}^{\mathrm{rb}\star}(t_{\dagger})$ is the ideal robot joint trajectory defined in \eqref{eq:ideal_robot}.

To obtain a slot-level fidelity measure, the execution-level errors falling within slot $t$ are accumulated over the associated execution indices
\begin{equation}
\label{eq:execution_index_set}
\mathcal{N}_{\dagger}(t)
=
\left\{
t_{\dagger}\in\mathbb{Z}_{+}:
t_{\dagger}\tau_{\dagger}\in[(t-1)\tau_s,\;t\tau_s)
\right\}.
\end{equation}
The reconstruction cost on platform $\dagger$ in slot $t$ is defined as
\begin{equation}
\label{eq:slot_joint_error}
\mathcal{E}_{\dagger}(t)
=\frac{1}{|\mathcal{N}_{\dagger}(t)|}\sum_{t_{\dagger}\in\mathcal{N}_{\dagger}(t)}
e_{\mathrm{joint}}(t_{\dagger}). 
\end{equation}

\subsection{Sim-to-Real Objective and Problem Formulation}
\label{subsec:problem}
Let $\mathcal{M}_{\mathrm{sim}}$ and $\mathcal{M}_{\mathrm{real}}$ denote the simulated and real teleoperation systems, respectively. The objective is to determine the sequence $\{\mathbf{f}(t)\}_{t\ge 1}$ that minimizes the average teleoperation cost on the real system. With the reconstruction cost $\mathcal{E}_{\mathrm{real}}(t)$ in \eqref{eq:slot_joint_error} and the effective communication energy cost $E(t)$ in \eqref{eq:energy_raw}, the real-system performance criterion is defined as
\begin{equation}
\label{eq:real_objective}
J_{\mathrm{real}}
=
\lim_{T\to\infty} \frac{1}{T}\, \mathbb{E}_{\mathcal{M}_{\mathrm{real}}}
\!\left[
\sum_{t=1}^{T}
\left(
\mathcal{E}_{\mathrm{real}}(t)
+
\lambda_E E(t)
\right)
\right],
\end{equation}

where $\lambda_E>0$ is a weighting coefficient that balances teleoperation accuracy and energy consumption, and $\mathbb{E}_{\mathcal{M}_{\mathrm{real}}}[\cdot]$ denotes the expectation with respect to the probability distribution induced by $\mathcal{M}_{\mathrm{real}}$. 

During policy training, the real teleoperation system $\mathcal M_{\rm real}$ is not treated as an interactive environment. Instead, it provides offline
real-domain trajectories, while the simulator $\mathcal M_{\rm sim}$ can be queried repeatedly to generate rollouts under candidate communication rate
policies. The discrepancy between the simulator-induced distribution and the offline real domain distribution is incorporated through a density-ratio
correction.
For each slot $t$, let $\chi_t$ denote the random variable induced by the teleoperation process. The probability distributions of $\chi_t$ under $\mathcal{M}_{\mathrm{real}}$ and $\mathcal{M}_{\mathrm{sim}}$ are denoted by $P_{\mathrm{real},t}$ and $P_{\mathrm{sim},t}$, respectively. \textcolor{black}{Both domains share the same robot kinematic model, action
space, and state representation, while their distributions may differ due
to sensing noise, actuation uncertainty, and unmodeled physical effects.}
Assuming $P_{\mathrm{real},t} \ll P_{\mathrm{sim},t}$, the Radon--Nikodym derivative of $P_{\mathrm{real},t}$ with respect to $P_{\mathrm{sim},t}$ exists\cite{folland1999real}. The density-ratio weight in slot $t$ is defined as
\begin{equation}
\label{eq:rt_def}
\iota_t
=
\frac{\mathrm{d} P_{\mathrm{real},t}}
     {\mathrm{d} P_{\mathrm{sim},t}}(\chi_t),
\end{equation}
which gives the importance weight for converting the expectation under $\mathcal{M}_{\mathrm{sim}}$ into that under $\mathcal{M}_{\mathrm{real}}$.

Accordingly, the sim-to-real problem is formulated as
\begin{align}
\label{eq:problem_p0}
\mathcal{P}0:\ \min_{\mathbf{f}(t),P(t)}
\ 
\lim_{T\to\infty}\frac{1}{T}&\,
\mathbb{E}_{\mathcal{M}_{\mathrm{sim}}}
\!\left[
\sum_{t=1}^{T}
\iota_t
\left(
\mathcal{E}_{\mathrm{sim}}(t)
+
\lambda_E E(t)
\right)
\right] \nonumber\\
\text{s.t.}\quad
&\eqref{fj_costrant}, \eqref{eq:payload_constraint}, \nonumber
\end{align}
where $\lambda_\mathrm{E}>0$ is a coefficient that calibrates the relative scale of the two terms in the unified objective, and $\mathbb{E}_{\mathcal{M}_{\mathrm{sim}}}[\cdot]$ denotes the expectation with respect to the probability distribution induced by $\mathcal{M}_{\mathrm{sim}}$. 

For any given $\mathbf{f}(t)$, the effective communication energy $E(t)$ in \eqref{eq:energy_raw} is monotonically increasing in $P(t)$. Therefore, the minimum energy required to support the transmission load $D(t)$ is attained when \eqref{eq:payload_constraint} holds with equality, 
\begin{equation}
\label{eq:power_min}
P_{\mathrm{req}}(t)
=\frac{N_0 B}{g(t)}
\left(
2^{\frac{D(t)}{B\tau_s}}-1
\right),
\end{equation}
and the corresponding minimum update energy cost is
\begin{equation}
\label{eq:energy_cost}
E_{\mathrm{req}}(t)
=
\zeta \tau_s \frac{N_0 B}{g(t)}
\left(
2^{\frac{D(t)}{B\tau_s}}-1
\right).
\end{equation}

Substituting \eqref{eq:power_min} and \eqref{eq:energy_cost} into $\mathcal{P}0$, an equivalent problem is obtained as
\begin{align}
\label{eq:final_problem}
\mathcal{P}1:\quad
\min_{\mathbf{f}(t)}
\quad
&
\lim_{T\to\infty}\frac{1}{T}\,
\mathbb{E}_{\mathcal{M}_{\mathrm{sim}}}
\!\left[
\sum_{t=1}^{T}
\iota_t
\left(
\mathcal{E}_{\mathrm{sim}}(t)
+
\lambda_E E_{\mathrm{req}}(t)
\right)
\right] \nonumber\\
\text{s.t.}\quad
&
f_j(t)\in\mathcal{F}_j,\quad \forall j,\ \forall t\ge 1. 
\end{align}

\subsection{Reinforcement Learning Reformulation}
\label{subsec:seq_reformulation}
The optimization problem $\mathcal{P}1$ is a stochastic sequential decision problem over an infinite horizon.
Due to the discrete sampling-rate decisions, the nonlinear communication energy model, and the coupled reconstruction cost induced by sampling, interpolation, and reconstruction, it is more suitably addressed through RL. Accordingly, $\mathcal{P}1$ is reformulated as a Markov decision process (MDP). \textcolor{black}{The policy $\pi$ is learned in simulation and deployed at the execution module for slot-wise decision making.}

\paragraph*{State} 
The latest available execution index before the beginning of slot $t$ on platform $\dagger$ is denoted by
\begin{equation}
\label{eq:latest_exec_index}
\bar t_{\dagger}(t)=\max\left\{t_{\dagger}\in\mathbb{Z}_{\ge0}: t_{\dagger}\tau_{\dagger}\le(t-1)\tau_s\right\}.
\end{equation}
Accordingly, the most recent reconstructed robot joint position available at slot $t$ is $\hat{\mathbf{q}}^{\mathrm{rb}}\!(\bar t_{\dagger}(t))$.
$\hat{\mathbf{v}}^{\mathrm{rb}}(t)$ and $\hat{\mathbf{a}}^{\mathrm{rb}}(t)$ denote the current reconstructed robot joint velocity and acceleration, respectively. The reconstructed robot joint velocity and acceleration are defined by 
\begin{equation}
\label{eq:state_velocity}
\hat{\mathbf{v}}^{\mathrm{rb}}(t)
=
\frac{
\hat{\mathbf{q}}^{\mathrm{rb}}\!\left(\bar t_{\dagger}(t)\right)
-
\hat{\mathbf{q}}^{\mathrm{rb}}\!\left(\bar t_{\dagger}(t)-1\right)
}{
\tau_{\dagger}
},
\end{equation}

\begin{equation}
\label{eq:state_acceleration}
\hat{\mathbf{a}}^{\mathrm{rb}}(t)
=
\frac{
\hat{\mathbf{q}}^{\mathrm{rb}}\!\left(\bar t_{\dagger}(t)\right)
-2\hat{\mathbf{q}}^{\mathrm{rb}}\!\left(\bar t_{\dagger}(t)-1\right)
+\hat{\mathbf{q}}^{\mathrm{rb}}\!\left(\bar t_{\dagger}(t)-2\right)
}{
\tau_{\dagger}^{2}
}.
\end{equation}
For the first two execution instants, the missing historical terms are initialized using the initial robot state.

To characterize the dimension-wise transmission load in the previous slot, define
\begin{equation}
\label{eq:state_load_vector}
\boldsymbol{\kappa}(t-1)
=
\tau_s\,\mathbf{d}\odot\mathbf{f}(t-1),
\end{equation}
where $\odot$ denotes the Hadamard product, and $\mathbf{f}(t-1)$ is the sampling-rate vector selected in the previous slot.

The state at slot $t$ is then defined as
\begin{equation}
\label{eq:mdp_state}
x_t=
\left[g(t),
\hat{\mathbf{q}}^{\mathrm{rb}}\!\left(\bar t_{\dagger}(t)\right)^\top,
\hat{\mathbf{v}}^{\mathrm{rb}}(t)^\top, 
\hat{\mathbf{a}}^{\mathrm{rb}}(t)^\top,
\boldsymbol{\kappa}(t-1)^\top
\right]^\top.
\end{equation}

\paragraph*{Action} The action at slot $t$ is the dimension-wise sampling-rate decision
\begin{equation}
\label{eq:mdp_action}
a_t=\mathbf{f}(t).
\end{equation}

\paragraph*{Reward} 
In the simulator $\mathcal{M}_{\mathrm{sim}}$, 
the instantaneous reward is defined as the negative cost
\begin{equation}
\label{eq:mdp_reward}
r_t=-\left(
\mathcal{E}_{\mathrm{sim}}(t)
+
\lambda_E E_{\mathrm{req}}(t)
\right).
\end{equation}
The density-ratio $\iota_t$ in $\mathcal{P}1$ is incorporated as an importance weight in Section \ref{sec:stage1}, rather than being absorbed into $r_t$.

\section{PAC-Bayes Generalization Analysis with Latent Density Ratio and Representation Bias}
\label{section III}
\subsection{Setup and Notation}
\label{PAC setup}
We consider a sim-to-real RL problem induced by the MDP reformulation in Section~\ref{subsec:seq_reformulation}. A parameterized 
encoder $\Omega_\theta$ maps the state $x_t$ in \eqref{eq:mdp_state} into a latent representation 
$z_t=\Omega_\theta(x_t)\in\mathcal{Z}$, and the policy operates in the latent space. For notational simplicity, we write $x\in\mathcal{X}$  and $a$ as generic
realizations of the state $x_t$ in \eqref{eq:mdp_state} and the action $a_t$ in \eqref{eq:mdp_action}, respectively. 
The simulator and the real system induce state distributions over
$\mathcal{X}$, denoted by $P_{\mathrm{sim}}$ and $P_{\mathrm{real}}$,
respectively. 
Since the encoder $\Omega_{\theta}$ maps state samples from both simulator and real domain into a common latent space, it is trained on mixed-domain samples to construct a shared latent representation across the two domains. 

\textcolor{black}{Let $\mathcal{H}$ be a hypothesis class, where each $h \in \mathcal{H}$ comprises the parameters of encoder $\Omega_\theta$, policy $\pi_\psi$, value function $V_\omega$,}
\begin{equation}
\begin{aligned}
h = (\theta, \psi, \omega).
\end{aligned}
\end{equation}

The simulator and real datasets are denoted by 
\begin{equation}
\begin{aligned}
\mathcal{D}_{\dagger}
\!=\!
\{ \tau^{\dagger}_i \}_{i=1}^{N_{\dagger}},
\ 
\tau^{\dagger}_i \!=\! (x^{\dagger}_{i,1},\dots,x^{\dagger}_{i,T_i^{\dagger}}),\ \dagger\!\in\!\{\mathrm{sim},\mathrm{real}\}.
\end{aligned}
\end{equation}
In the dataset $\mathcal{D}_{\dagger}$, $\tau^{\dagger}_i$ is the $i$-th trajectory, $N_{\dagger}$ is the number of trajectories, and
$T_i^{\dagger}$ denotes the length of the $i$-th trajectory.

Define the simulator and real sample sets as
\begin{equation}
\label{sampleset}
    \mathcal{X}_{\dagger}\!
=\!
\{ x^{\dagger}_{i,t} \!\mid\! 1 \!\le\! i \!\le\! N_{\dagger},1\!\le\! t \!\le\! T_i^{\dagger} \},\ \dagger\!\in\!\{\mathrm{sim},\mathrm{real}\},
\end{equation}
with size $n_{\dagger}\!=\!|\mathcal{X}_{\dagger}|$. Here, $t$ denotes the decision-slot index, consistent with Section~\ref{sec:system_model}.

To separate representation warm-up from density-ratio estimation, the simulator samples are divided into a warm-up set $\mathcal S_h$ and a ratio-estimation set $\mathcal S_\iota$
\begin{equation}
    \mathcal{S}_h \cup \mathcal{S}_\iota = \mathcal{X}_{\mathrm{sim}}, 
\ \mathcal{S}_h \cap \mathcal{S}_\iota = \varnothing.
\end{equation}
The encoder warm-up is performed on $\mathcal{S}_h \cup \mathcal{X}_{\mathrm{real}}$,
while density-ratio estimation uses $\mathcal{S}_\iota$ and $\mathcal{X}_{\mathrm{real}}$.

\subsection{Population Risks on Real Domain}
The population risk of a hypothesis $h\in\mathcal{H}$ under distribution $P_{\mathrm{real}}$ is defined as
\begin{equation}
\label{eq:population_risk}
R(h)
=
\mathbb{E}_{x\sim P_{\mathrm{real}}}
\bigl[
L(h,x)
\bigr],
\end{equation}
where $L(h,x)\in[0,1]$ is a bounded loss.

 The slot-level density-ratio weight $\iota_t$ in \eqref{eq:rt_def} is defined on $\chi_t$ for the objective formulation.
For the generalization analysis, we use the real-to-simulator
density ratio over $\mathcal{X}$
\begin{equation}
\label{eq:real_to_sim_ratio}
\iota_{\mathcal{X}}^{\star}(x)
=
\frac{\mathrm{d}P_{\mathrm{real}}}
     {\mathrm{d}P_{\mathrm{sim}}}(x).
\end{equation}
When $x=x_t$, this ratio is evaluated at the state sample generated in
decision slot $t$.

We assume that $P_{\mathrm{real}}$ is absolutely continuous with respect to $P_{\mathrm{sim}}$, denoted by $P_{\mathrm{real}}\ll P_{\mathrm{sim}}$. This assumption means that any event in $\mathcal{X}$ with zero probability under $P_{\mathrm{sim}}$ also has zero probability under $P_{\mathrm{real}}$. It guarantees the existence of the Radon--Nikodym derivative in \eqref{eq:real_to_sim_ratio}. 

By construction, the real-domain risk can be rewritten as
\begin{equation}
\label{eq:real-risk-ratio}
R(h)
=\mathbb{E}_{x\sim P_{\mathrm{sim}}}\big[\iota_{\mathcal{X}}^\star(x)L(h,x)\big].
\end{equation}

Since the policy operates in latent space, for a given encoder
$\Omega_\theta$, we consider the latent distributions
$P_{\mathrm{sim}}^\theta=P_{\mathrm{sim}}\circ\Omega_\theta^{-1}$
and
$P_{\mathrm{real}}^\theta=P_{\mathrm{real}}\circ\Omega_\theta^{-1}$.
The corresponding latent density ratio is denoted by
\begin{equation}
\iota_{\theta}^\star(z)
=
\frac{dP_{\mathrm{real}}^\theta}{dP_{\mathrm{sim}}^\theta}(z),
\ z=\Omega_\theta(x).
\end{equation}
In implementation, an estimator $\hat{\iota}(z)$ is fitted in the latent space using kernel-based density-ratio estimation methods, e.g., uLSIF and KLIEP.
For notational simplicity, we omit the dependence of $\hat\iota$ on the encoder parameter $\theta$ when the encoder is clear from the context.

\subsection{Importance-Weighted PAC--Bayes Deviation Analysis}
\label{III-C}
After the density-ratio estimator $\hat{\iota}$ is fitted using $\mathcal S_\iota$ and $\mathcal X_{\mathrm{real}}$ in the frozen
reference latent space, it is treated as fixed in the subsequent deviation analysis. 
If the raw estimator takes negative values, it is clipped to be nonnegative before normalization, and we still use $\hat{\iota}$ to denote the resulting nonnegative estimator in the deviation analysis.
The analysis is interpreted for a fixed encoder used to construct the latent space at the corresponding training stage; the dependence on the encoder is captured by $\varepsilon_\iota(\theta,\hat{\iota})$ and $\varepsilon_{\mathrm{enc}}(\theta)$.
Let $\mathcal B_{\mathrm{sim}}$ denote a simulator rollout batch collected after fitting $\hat{\iota}$. 
The estimated weights are normalized to have unit mean over $\mathcal B_{\mathrm{sim}}$, and we reuse $\hat{\iota}$ to denote the normalized estimator.
The real-domain risk~\eqref{eq:real-risk-ratio} is approximated by 

\begin{equation}
\label{eq:importance-weighted-real}
\widehat R(h)
=
\frac{1}{|\mathcal B_{\mathrm{sim}}|}
\sum_{x\in\mathcal B_{\mathrm{sim}}}
\hat{\iota}(\Omega_\theta(x))L(h,x).
\end{equation}

Following standard importance-weighted analysis, we decompose

\begin{equation}
\begin{aligned}
R(h)
= &
\mathbb{E}_{x\sim P_{\mathrm{sim}}}[\hat{\iota}(\Omega_\theta(x))L(h,x)]
+\\
&
\mathbb{E}_{x\sim P_{\mathrm{sim}}}[(\iota_{\mathcal{X}}^\star(x)-\hat{\iota}(\Omega_\theta(x)))L(h,x)].
\end{aligned}
\end{equation}
The density-ratio estimation error is defined as
\begin{equation}
\label{eq:err-ratio-final}
\varepsilon_\iota(\theta,\hat{\iota})=\sup_{h\in\mathcal{H}}
\left|
\mathbb{E}_{x\sim P_{\mathrm{sim}}}[(\iota_{\mathcal{X}}^\star(x)-\hat{\iota}(\Omega_\theta(x)))L(h,x)]
\right|.
\end{equation} 

For a data-dependent posterior $Q$ over $\mathcal H$, define
\begin{equation}
R(Q)
=
\mathbb E_{h\sim Q}[R(h)],
\ 
\widehat R(Q)
=
\mathbb E_{h\sim Q}[\widehat R(h)].
\end{equation}

Consider a prior $P$ that is chosen independently of the simulator samples used in the empirical risk estimate, and a data-dependent posterior $Q$ over $\mathcal H$.
For any confidence parameter $\delta\in(0,1)$, let $C>0$ be a universal constant appearing in PAC--Bayes inequalities\cite{pac}. The standard importance-weighted PAC-Bayes analysis gives the following deviation
\begin{equation}
\begin{aligned}
\label{eq:iw-bound-real}
|
\mathbb E_{h\sim Q}
\mathbb E_{x\sim P_{\mathrm{sim}}}
[
\hat{\iota}(\Omega_\theta(x))L(h,&x)
]
-
\widehat R(Q)
|
\lesssim\\
&\sqrt{
\frac{
\mathrm{KL}(Q\|P)+\ln(C/\delta)
}{
2n_{\mathrm{eff}}(\hat{\iota})
}
}.
\end{aligned}
\end{equation}
Here, the effective sample size is computed from the normalized weights as
\begin{equation}
\label{eq:effective-sample-size}
n_{\mathrm{eff}}(\hat{\iota})
=
\frac{
\left(
\sum_{x\in\mathcal B_{\mathrm{sim}}}
\hat{\iota}(\Omega_\theta(x))
\right)^2
}{
\sum_{x\in\mathcal B_{\mathrm{sim}}}
\hat{\iota}^2(\Omega_\theta(x))
}.
\end{equation}

Combining \eqref{eq:importance-weighted-real}--\eqref{eq:effective-sample-size} and introducing a representation discrepancy associated with the learned encoder, we obtain the following real-domain generalization characterization
\begin{equation}
\label{eq:final-bound-real-final}
\begin{aligned}
R(Q)
\lesssim
&\widehat{R}(Q)
+
\varepsilon_\iota(\theta,\hat{\iota})\\
&+
\sqrt{
\frac{
\mathrm{KL}(Q\|P)+\ln(C/{\delta})
}{
2n_{\mathrm{eff}}(\hat{\iota})
}}+
\varepsilon_{\mathrm{enc}}(\theta).
\end{aligned}
\end{equation}
 where $\varepsilon_{\mathrm{enc}}(\theta)$ denotes the representation discrepancy induced by using the learned encoder $\Omega_\theta$ in place of an ideal encoder.

\subsection{Encoder Bias}
Let $\Omega^\star$ denote an ideal encoder under which the real-to-simulator
density ratio is sufficiently smooth and well approximated in the latent space. Define the encoder discrepancy as
\begin{equation}
\label{eq:encoder_discrepancy}
\epsilon_\Omega(\theta)
=
\mathbb E_{x\sim P_{\mathrm{real}}}
[
\|\Omega_\theta(x)-\Omega^\star(x)\|
].
\end{equation}
Assume that the loss $L$ is $L_z$–Lipschitz with respect to the latent representation. Then, by a standard Lipschitz argument,
\begin{equation}
\label{eq:encoder_var}
\varepsilon_{\mathrm{enc}}(\theta)
\le
L_z\epsilon_\Omega(\theta).
\end{equation}

Substituting \eqref{eq:encoder_discrepancy} into \eqref{eq:final-bound-real-final}, we obtain the following instantiated real-domain generalization characterization
\begin{equation}
\label{eq:pretrain-bound-final}
\begin{aligned}
R(Q)
\lesssim&
\widehat{R}(Q)
+
\varepsilon_{\iota}(\theta,\hat{\iota})
\\& +
\sqrt{
\frac{
\mathrm{KL}(Q\|P)+\ln(C/{\delta})
}{
2n_{\mathrm{eff}}(\hat{\iota})
}}+
L_z\,\epsilon_\Omega(\theta).
\end{aligned}
\end{equation}
This characterization identifies four contributing terms to the real-domain generalization error:

(i) the importance-weighted risk of the real domain $\widehat{R}(Q)$;

(ii) density-ratio approximation error $\varepsilon_{\iota}(\theta,\hat{\iota})$, induced by approximating the target ratio $\iota_{\mathcal{X}}^\star$ with the latent-space estimator $\hat{\iota}(\Omega_\theta(x))$; 

(iii) the PAC--Bayes deviation term, which captures finite-sample uncertainty through $\mathrm{KL}(Q\|P)$ and $n_{\mathrm{eff}}(\hat{\iota})$;

(iv) the encoder representation bias $L_z\epsilon_\Omega(\theta)$ induced by the mismatch between $\Omega_\theta$ and $\Omega^\star$.

\section{Sim-to-Real Domain-Adaptive Policy Optimization}
Building upon the PAC-Bayes analysis, we design a structured training framework that addresses the density-ratio estimation error and the encoder representation term identified above. The overall procedure consists of three phases: an encoder warm-up phase followed by two training stages. 
\begin{enumerate}
    \item Warm-up: The encoder is initialized on
    $\mathcal{S}_h \cup \mathcal{X}_{\mathrm{real}}$ by minimizing an MMD based warm-up objective. This phase provides a shared latent representation for simulator and real-domain samples.  
    \item Stage 1: With the encoder frozen at $\theta_0$, a latent real-to-simulator density ratio $\hat{\iota}$ is estimated and used as the importance weight. The communication-rate policy is then optimized in the latent state space using PPO.
    \item Stage 2: The encoder is unfrozen and fine-tuned with a trust-region regularization centered at $\theta_0$ and a latent alignment objective based on weighted MMD. During this stage, the policy
and value networks are kept fixed, and only the encoder parameter \(\theta\) is updated. After encoder fine-tuning, the reference encoder
is updated by setting \(\theta_0\leftarrow\theta\), and policy optimization continues in the next outer iteration.
\end{enumerate}

The illustration of algorithm
architecture is shown in Fig. \ref{algorithm}. The warm-up phase initializes a coherent latent space. Stage~1 fixes the encoder to enable stable density-ratio estimation and reliable policy optimization in the latent space. Stage~2 then allows controlled encoder refinement through trust-region regularization, aiming to improve cross-domain alignment while keeping the latent representation close to the Stage~1 reference space. Together, these stages separate representation learning into a stable initialization phase and a subsequent regulated adaptation phase. The following subsections detail each component of the training framework.

\begin{figure*}[t]
\centering{\includegraphics[width=1\textwidth]{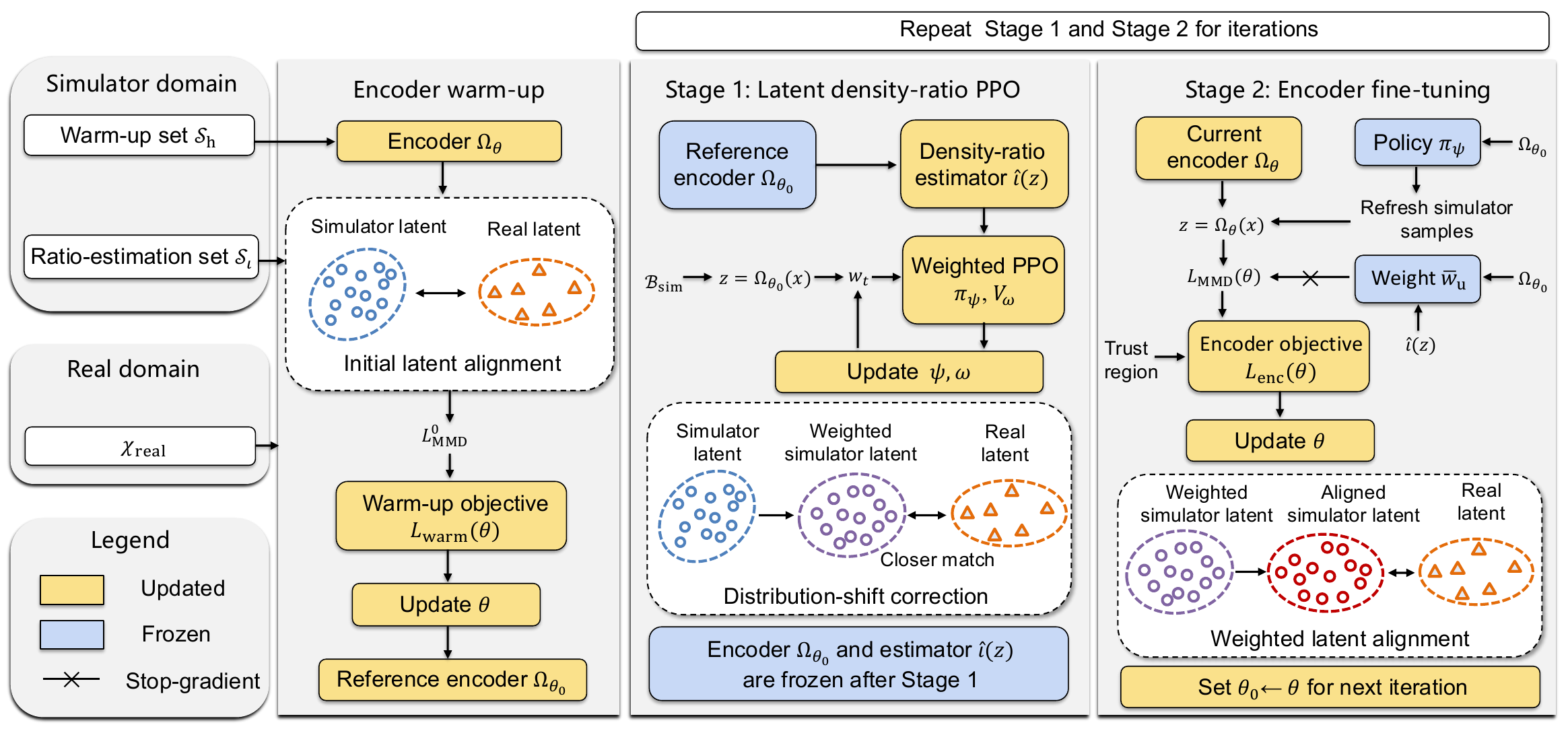}}
\caption{Illustration of proposed sim-to-real domain-adaptive policy optimization algorithm architecture. }
\label{algorithm}
\end{figure*}

\subsection{Encoder Warm-up}
\label{Encoder Warm-up}
We employ a multi-layer perceptron (MLP) as the encoder $\Omega_\theta$.
Given simulator and real-domain minibatches
\(\{\check{x}_u^{\mathrm{sim}}\}_{u=1}^{n_\mathrm{w,s}}\subset\mathcal{S}_h\)\ and
\(\{\check{x}_v^{\mathrm{real}}\}_{v=1}^{n_\mathrm{w,r}}\subset\mathcal{X}_{\mathrm{real}}\), their latent representations
are defined as
\begin{equation}
    \check{z}_u^{\mathrm{sim}}=\Omega_\theta(\check{x}_u^{\mathrm{sim}}),
    \quad
    \check{z}_v^{\mathrm{real}}=\Omega_\theta(\check{x}_v^{\mathrm{real}}).
\end{equation}
We use an MMD loss to reduce the initial discrepancy between
the simulator and real-domain latent distributions:

\begin{align}
\label{eq:warmup_mmd}
L_{\mathrm{MMD}}^{0}(\theta)
=&
\frac{1}{{n_\mathrm{w,s}}^2}
\!\sum_{u, u^{\prime}}\!
\varpi\!
\left(
\check{z}_u^{\mathrm{sim}}\!,
\check{z}_{u'}^{\mathrm{sim}}
\right)
\!-\!
\frac{2}{n_\mathrm{w,s}n_\mathrm{w,r}}\!
\sum_{u, v}\!
\varpi\!
\left(
\check{z}_u^{\mathrm{sim}}\!,
\check{z}_v^{\mathrm{real}}
\right)\nonumber
\\
&+
\frac{1}{{n_\mathrm{w,r}}^2}
\sum_{v,v^{\prime}}
\varpi
\left(
\check{z}_v^{\mathrm{real}},
\check{z}_{v'}^{\mathrm{real}}
\right),
\end{align}
where $\varpi(z,z')$ is a Gaussian mixture kernel defined on the latent space,
\begin{equation}\label{varpi}
    \varpi(z,z') = \sum_{\ell=1}^{L} \alpha_{\ell} \exp\!\left( 
    -\frac{\|z - z'\|^{2}}{2\sigma_{\ell}^{2}}
\right),
\end{equation}
where $\sigma_\ell>0$ controls the kernel width and
$\alpha_\ell>0$ denotes the mixture weight satisfying
$\sum_{\ell=1}^L\alpha_\ell=1$.

The warm-up objective is
\begin{equation}
    L_{\mathrm{warm}}(\theta)
    =
    \lambda_{\mathrm{MMD}}^{0}
    L_{\mathrm{MMD}}^{0}(\theta)
    +\lambda_{\mathrm{reg}}\|\theta\|_2^2,
\label{eq:warmup_objective}
\end{equation}
where \(\lambda_{\mathrm{MMD}}^{0},\lambda_{\mathrm{reg}}\) are nonnegative coefficients and $\|\theta\|_2^2$ is a regularization term. The warm-up on $\mathcal S_h\cup\mathcal X_{\mathrm{real}}$ provides a coherent latent representation across the simulator and the real domains. After optimizing \(L_{\mathrm{warm}}(\theta)\) in the warm-up phase, 
we
denote the encoder parameter by \(\theta_0\), and the
initialized encoder by \(\Omega_{\theta_0}\).

\subsection{Latent-Space PPO}
\label{sec:stage1}
In Stage 1, the general characterization in Section \ref{section III} is instantiated with the frozen encoder $\Omega_{\theta_0}$. All simulator and real samples are first mapped into the corresponding reference latent space. The simulator and real latent sample sets are defined as
\begin{equation}
    \mathcal{Z}_{\dagger}
=
\big\{ z^{\dagger}_{i,t} = \Omega_{\theta_0}(x^{\dagger}_{i,t}) \big\},
\ \dagger\in\{\mathrm{sim},\mathrm{real}\},
\end{equation}
where $i$ indexes trajectories and $t$ denotes the decision slot within each trajectory.

To optimize the sim-to-real control objective in $\mathcal{P}1$, a latent density ratio $\hat {\iota}(z)$ is estimated
\begin{equation}
\label{eq:latent_real_ratio_direct}
\hat{\iota}(z)
\approx
\frac{
\mathrm{d}P_{\mathrm{real}}^{\theta_0}
}{
\mathrm{d}P_{\mathrm{sim}}^{\theta_0}
}(z),
\end{equation}
where $P_{\mathrm{sim}}^{\theta_0}$ and $P_{\mathrm{real}}^{\theta_0}$ are the latent distributions induced by the frozen encoder $\Omega_{\theta_0}$. The estimator $\hat{\iota}(z)$ is fitted using latent samples by uLSIF or KLIEP\cite{uLSIF,KLIEP}. Detailed descriptions 
of uLSIF and KLIEP are provided in the Appendix.

For each iteration, simulator trajectories are collected as $(x_1, a_1, r_1, \ldots, x_{T})$, where the horizon $T$ may vary across trajectories.  
Each state is mapped to latent space $z_t = \Omega_{\theta_0}(x_t)$.
Generalized advantage estimation (GAE) is computed along the trajectory:
\begin{equation}\label{PPO_delta}
    \delta_t 
= r_t + \gamma V_\omega(z_{t+1}) - V_\omega(z_t), 
\end{equation}

\begin{equation}\label{PPO advantage}
    A_t = \sum_{\ell=0}^{T-t} (\gamma\lambda_{\mathrm{GAE}})^\ell \,\delta_{t+\ell}.
\end{equation}

Following the convention in Section~\ref{III-C}, $\hat{\iota}$ denotes the nonnegative density-ratio estimator after clipping when necessary.
The real-to-simulator density ratio provides the normalized importance weight
\begin{equation}\label{PPO important weight}
    w_t = \frac{\hat \iota(z_t)}{\mathbb{E}_t[\hat \iota(z_t)]},
\end{equation}
where $\mathbb{E}_t[\cdot]$ denotes the average over the rollout batch.

The policy and value networks are updated by minimizing the following weighted PPO loss:
\begin{equation}
\begin{aligned}
\label{PPO loss}
    L_{\mathrm{PPO}}(\psi,\omega)
=L_{\mathrm{pol}}(\psi)-\mu_{\mathrm{ent}}L_{\mathrm{ent}}(\psi)+\mu_{\mathrm{val}}L_{\mathrm{val}}(\omega),
\end{aligned}
\end{equation}
where $\mu_{\mathrm{ent}}$ and $\mu_{\mathrm{val}}$ are positive constants; the clipped policy loss is
\begin{equation}
    L_{\mathrm{pol}}(\psi)=-\mathbb{E}_{t}\!\left[
w_t\,
\min\!\big(
\rho_t A_t,
\operatorname{clip}(\rho_t,1-\epsilon,1+\epsilon)A_t
\big)
\right],
\end{equation}
with the importance sampling ratio
\begin{equation}\label{sampling ratio}
    \rho_t 
=
\frac{\pi_\psi(a_t \mid z_t)}
     {\pi_{\psi_{\text{old}}}(a_t \mid z_t)};
\end{equation}
the entropy regularization term is defined as
\begin{equation}
L_{\mathrm{ent}}(\psi)=-\mathbb{E}_t\!\left[\sum_{a_t}
(\pi_\psi(a_t|z_t)\log\pi_\psi(a_t|z_t))
\right];
\end{equation}
the value function term is
\begin{equation}
    L_{\mathrm{val}}(\omega)=\mathbb{E}_t\!\left[
w_t\left(V_\omega(z_{t})- G_t\right)^2
\right],
\end{equation}
where $ G_t=A_t+V_{\omega_{\mathrm{old}}}(z_t)$ denotes the value target constructed from generalized advantage estimation (GAE) \cite{GAE}.

Since the encoder remains fixed at $\theta_0$, the representation bias term $L_z\epsilon_\Omega(\theta_0)$ is fixed, and the density-ratio estimation is performed in a fixed latent space. Thus, the learned posterior $Q_{\text{stage1}}$ follows the characterization in \eqref{eq:pretrain-bound-final}: 
\begin{equation}
\label{stage1 PAC}
\begin{aligned}
R(&Q_{\text{stage1}})\lesssim
\widehat R(Q_{\text{stage1}})
+
\varepsilon_\iota(\theta_0,\hat \iota)
\\& +
\sqrt{
\frac{
\mathrm{KL}(Q_{\text{stage1}}\|P)
+
\ln(C/\delta)
}{
2n_{\mathrm{eff}}(\hat \iota)
}}
+
L_z\,\epsilon_\Omega(\theta_0).
\end{aligned}
\end{equation}

\subsection{Trust-Region Encoder Fine-tuning}
\label{sec:IV_C}
In Stage 2, the encoder 
parameters are unfrozen. To prevent uncontrolled drift in the latent space, which would compromise both policy optimization and density-ratio estimation, we introduce an explicit trust-region regularization. Specifically, a quadratic penalty centered at $\theta_0$ is imposed, where the coefficient $\beta>0$ controls the strength of the trust region,
\begin{equation}
\label{eq:stage2-obj}
\min_{\theta} 
L_{\mathrm{enc}}(\theta)=\lambda_{\mathrm{MMD}}\,L_{\mathrm{MMD}}(\theta)
+ \beta \|\theta - \theta_0\|_2^2.
\end{equation}
Within each outer iteration, $\theta_0$ is kept fixed and serves as the trust-region anchor. Here, $L_{\mathrm{MMD}}(\theta)$ measures the discrepancy between the weighted simulator latent samples and the real latent samples. 

To keep encoder fine-tuning data consistent with the latest policy, we refresh simulator-state minibatches by rolling out $\mathcal{M}_{\rm sim}$ with
the frozen policy
$\pi_\psi(\cdot|\Omega_{\theta_0}(x))$ obtained from Stage 1. Given
minibatches $\{x_u^{\rm sim}\}_{u=1}^{n_{t,s}}$ and
$\{x_v^{\rm real}\}_{v=1}^{n_{t,r}}$, the latent representations are
\begin{equation}
z_u^{\mathrm{sim}}=\Omega_\theta(x_u^{\mathrm{sim}}),
\ 
z_v^{\mathrm{real}}=\Omega_\theta(x_v^{\mathrm{real}}).
\end{equation}
For the simulator minibatch, the stabilized weight is computed in the reference latent space using the encoder $\Omega_{\theta_0}$, under which the density-ratio estimator was fitted
\begin{equation}
\label{eq:normalized_mmd_weight}
\bar w_u =
\frac{
\hat{\iota}(\Omega_{\theta_0}(x^{\mathrm{sim}}_u))
}{
\max\left\{
\sum_{s=1}^{n_{\mathrm{t, s}}}
\hat{\iota}(\Omega_{\theta_0}(x_s^{\mathrm{sim}})),
\epsilon_w
\right\}
},
\end{equation}
where $\epsilon_w>0$ is used only to avoid numerical degeneracy when all estimated weights are close to zero.
The estimator $\hat{\iota}$ is fixed after Stage~1, while the kernel evaluations in \(L_{\mathrm{MMD}}(\theta)\) use the current encoder $\Omega_\theta$. Thus, \(L_{\mathrm{MMD}}(\theta)\) can update the encoder while the importance weights remain tied to the reference latent space.
\begin{align}
\label{MMD2}
L_{\mathrm{MMD}}(\theta)
=
&\sum_{u,u'}
\bar w_u\bar w_{u'}
\varpi\!(z_u^{\mathrm{sim}}\!,z_{u'}^{\mathrm{sim}})\!-\!
\frac{2}{n_{\mathrm{t, r}}}\!
\sum_{u,v}\!
\bar w_u
\varpi\!(z_u^{\mathrm{sim}}\!,z_v^{\mathrm{real}})
\nonumber\\
&+
\frac{1}{{n_{\mathrm{t, r}}}^2}
\sum_{v,v'}
\varpi\!\left(z_v^{\mathrm{real}}\!,z_{v'}^{\mathrm{real}}\right).
\end{align}

The policy and encoder optimization are intentionally decoupled to preserve training stability. In Stage 1, the weighted PPO objective in (\ref{PPO loss}) updates $(\psi,\omega)$, under the frozen reference encoder \(\Omega_{\theta_0}\). In Stage~2, $\pi_\psi$ is a frozen policy with input
$\Omega_{\theta_0}(x)$ to refresh simulator-state minibatches. The objective in (\ref{eq:stage2-obj}) is then optimized using the current encoder $\Omega_\theta$ for the MMD kernel evaluations, while the importance
weights remain computed in the reference latent space. Therefore,
no PPO gradient is propagated into the encoder during Stage~2. This separation avoids mutual interference between policy optimization and latent-space alignment. 

We now establish a stability result showing that bounded encoder drift introduces a controlled perturbation to the real-domain risk. Since Stage 2 freezes the policy and value networks, the posterior variation in this stage is induced only by the encoder parameter
$\theta$, while $\psi$ and $\omega$ are treated as fixed components
inherited from Stage 1.

\begin{assumption}[Lipschitz continuity with respect to encoder parameters]
\label{assump:lipschitz-encoder}
For any two encoders $\Omega_{\theta}$ and $\Omega_{\theta_0}$ satisfying
$\|\theta-\theta_0\|_2\le\varkappa$, and for any policy parameter $\psi$ and
value parameter $\omega$, the loss function satisfies
\begin{equation}
\label{eq:lipschitz_encoder_assumption}
\big|
L((\theta,\psi,\omega),x)
-
L((\theta_0,\psi,\omega),x)
\big|
\le
L_\theta\|\theta-\theta_0\|_2,
\ 
\forall x\sim P_{\mathrm{real}}.\nonumber
\end{equation}
This condition holds, for instance, when the loss is Lipschitz continuous with respect to the latent representation and the encoder is Lipschitz continuous with respect to its parameters.
\end{assumption}

\begin{proposition}[Stability under bounded encoder drift]
\label{prop:stage2-stability}
Let $Q_{\mathrm{stage2}}$ denote the posterior after Stage~2, and let
$Q_{\mathrm{stage2}}^{0}$ denote the posterior obtained from
$Q_{\mathrm{stage2}}$ by replacing the encoder parameter $\theta$ with
$\theta_0$ while keeping $(\psi,\omega)$ unchanged.
Suppose that Stage~2 satisfies
$\|\theta-\theta_0\|_2\le\varkappa$ almost surely under
$Q_{\mathrm{stage2}}$. Then, under Assumption~\ref{assump:lipschitz-encoder},
\begin{equation}
\label{eq:stage2-stability-corrected}
R(Q_{\mathrm{stage2}})
\le
R(Q_{\mathrm{stage2}}^{0})
+
L_\theta\varkappa .
\end{equation}
\end{proposition}

\begin{IEEEproof}
For any $h=(\theta,\psi,\omega)$ sampled from $Q_{\mathrm{stage2}}$, define
$h_0=(\theta_0,\psi,\omega)$ by replacing the encoder parameter $\theta$ with
$\theta_0$ while keeping $(\psi,\omega)$ unchanged.
By Assumption~\ref{assump:lipschitz-encoder},
\begin{equation}
\big|
L(h,x)-L(h_0,x)
\big|
\le
L_\theta\|\theta-\theta_0\|_2
\le
L_\theta\varkappa .\nonumber
\end{equation}
Taking expectation over $x\sim P_{\mathrm{real}}$ and
$h\sim Q_{\mathrm{stage2}}$ gives \eqref{eq:stage2-stability-corrected}.
\end{IEEEproof}

Proposition~\ref{prop:stage2-stability} justifies the trust-region design in \eqref{eq:stage2-obj}. 
Stage~2 allows the encoder to adapt through the MMD objective, while the proximal penalty $\beta\|\theta-\theta_0\|_2^2$ discourages large deviations from the Stage~1 encoder. 
This design limits abrupt latent-space drift that could invalidate the learned density-ratio estimator and destabilize policy optimization. 
In practice, $\beta$ is gradually decayed across outer iterations, enabling conservative encoder updates at the early stage and increased flexibility after the policy and density-ratio estimator have stabilized. 

After Stage~2, the reference encoder is updated by setting \(\theta_0\leftarrow\theta\), and the next iteration re-estimates the latent density ratio and continues weighted PPO training under the updated reference encoder.
The overall procedure of the proposed algorithm is presented in Algorithm~\ref{alg:two-stage-rl}.

\begin{algorithm}[t]
\caption{Sim-to-Real Domain-Adaptive Policy Optimization}
\label{alg:two-stage-rl}
\begin{algorithmic}[1]
\STATE \textbf{Initialize:}
Interactive simulator environment $\mathcal{M}_{\mathrm{sim}}$; datasets $\mathcal{D}_{\mathrm{sim}},\mathcal{D}_{\mathrm{real}}$; simulator-sample split $(\mathcal{S}_h,\mathcal{S}_\iota)$; warm-up steps $K_0$; global iterations $K_1$; Stage 1 iterations $K_2$; Stage 2 iterations $K_3$;
trust-region weight $\beta_0$; decay factor $\mu$; regularization coefficient $\lambda_{\mathrm{reg}}$; MMD coefficient $\lambda_{\mathrm{MMD}}^0$; weighted MMD coefficient $\lambda_{\mathrm{MMD}}$.

\STATE \textbf{Warm-up:}\\
\FOR{$k_0=1\cdots K_0$}
   \STATE Sample warm-up minibatches from
    \(\mathcal S_h\) and \(\mathcal X^{\mathrm{real}}\).
    \STATE Update encoder parameter \(\theta\) by minimizing \eqref{eq:warmup_objective}.
\ENDFOR
\STATE Set initial encoder parameter \(\theta_0\leftarrow\theta\).
\STATE Initialize  PPO policy and value parameters $(\psi,\omega)$.
\STATE Set $\beta\gets \beta_0$.

\FOR{$k_1=1\cdots K_1$}

    \STATE \textbf{Stage 1. Frozen encoder with latent density-ratio estimation and PPO:}
    \STATE Freeze encoder at $\theta\gets\theta_0$. 
    \STATE Estimate latent real-to-simulator density ratio $\hat \iota(\cdot)$ using uLSIF or KLIEP in the Appendix.
    
    \FOR{ $k_2=1\cdots K_2$}
    \STATE Roll out simulator trajectories in $\mathcal{M}_{\mathrm{sim}}$ using $\pi_\psi(\cdot\mid z)$ with $z=\Omega_{\theta_0}(x)$, and collect the batch $\mathcal B_{\mathrm{sim}}$.
        \STATE Compute GAE advantages $A_t$ and real-to-simulator importance weights $w_t$ by \eqref{PPO_delta}--\eqref{PPO important weight}.
        \STATE Update $(\psi,\omega)$ by minimizing the weighted PPO objective \eqref{PPO loss} with density ratio \eqref{sampling ratio}.
    \ENDFOR

    \STATE \textbf{Stage 2. Encoder fine-tuning with trust region and weighted MMD:}
    \STATE Unfreeze encoder $\theta$, while keeping \(\theta_0\),
\(\hat{\iota}\), \(\psi\), and \(\omega\) fixed.
    \FOR{ $k_3=1\cdots K_3$}
         \STATE Roll out simulator trajectories in $\mathcal{M}_{\mathrm{sim}}$ using the policy $\pi_\psi(\cdot\mid z)$, where $z=\Omega_{\theta_0}(x)$. Collect simulator minibatches $\{x^{\rm sim}_u\}_{u=1}^{n_\mathrm{t,s}}$, and collect  real-domain minibatches $\{x^{\rm real}_v\}_{v=1}^{n_\mathrm{t,r}}$ from $\mathcal{X}^{\mathrm{real}}$.
\STATE Encode them using the current encoder as $z^{\mathrm{sim}}_u=\Omega_\theta(x^{\mathrm{sim}}_u)$, $z^{\mathrm{real}}_v=\Omega_\theta(x^{\mathrm{real}}_v)$.
    \STATE Compute the stabilized weights \(\bar w_u\) using the frozen \(\Omega_{\theta_0}\) and $\hat{\iota}$ according to \eqref{eq:normalized_mmd_weight}.
    \STATE Compute the weighted MMD loss in \eqref{MMD2}.
        \STATE Update the encoder parameter $\theta$ by minimizing \eqref{eq:stage2-obj}. 
    \ENDFOR
    \STATE Set $\theta_0\gets\theta$.
    \STATE Decay trust-region weight $\beta \gets \mu \beta$.
\ENDFOR

\STATE \textbf{return} encoder $\Omega_\theta$ and policy $\pi_\psi$.
\end{algorithmic}
\end{algorithm}

\section{Simulation Results}\label{simulation_results}
\subsection{Environment Settings}
Since human motion states in XR teleoperation usually evolve continuously over a second-level time scale, we adopt a relatively low-frequency decision frequency to reduce the additional overhead induced by receiving dimension-wise control at the XR device and avoid undesirable control oscillations induced by frequently sampling-rate switching.
Specifically, the decision frequency is 5 Hz, i.e., $\tau_s=200$ ms. Motivated by high-rate XR tracking modes in commercial XR devices, $f_j^{\max}=120$ Hz and $\mathcal{F}_j=\{0,30,60,120\}$ Hz. The execution frequencies of the simulated and real platforms are set to $f_{\mathrm{sim}}=120$ Hz and $f_{\mathrm{real}}=30$ Hz, i.e., $\tau_{\mathrm{sim}}=8.3$ ms,  $\tau_{\mathrm{real}}=33.3$ ms, $T=2000$. All experiments were performed on an NVIDIA RTX 3070 GPU and Ubuntu 22.04.

\subsubsection{Communication parameter setup} Following the 3GPP TR 38.901 indoor-office channel model\cite{3GPP_mmwave}, the path-loss model supports both LOS and NLOS conditions, with the user mobility set to 3 km/h. For room-scale indoor XR teleoperation, we consider the distance in $[3,30]$ m and $f_c=7$ GHz. To account for shadow fading and occasional human-body or device blockage, we generate the channel gain $10\lg g(t)\in[-105,-55]\ {\rm dB}$. $N_0$ is set to -174 dBm/Hz, $B=100$ MHz. $\zeta=3\times10^4$ and $\lambda_\mathrm{E}=5$.

\subsubsection{Real-world and simulator platform}
The real-domain dataset $\mathcal{D}_{\mathrm{real}}$ is constructed based on Humanoid Everyday, a comprehensive public dataset that contains motion data collected by Apple Vision Pro\cite{apple_vision_pro} and humanoid robot data recorded in real-world task scenarios\cite{dataset}. For $\mathcal{D}_{\mathrm{sim}}$, we develop a customized simulation environment named wireless\_xr\_teleoperate based on the unitree\_sim\_isaaclab in the Unitree toolkit xr\_teleoperate\cite{Unitree}. The underlying simulator is built on NVIDIA Isaac Sim 4.5.0 and Isaac Lab. To ensure that the simulator in unitree\_sim\_isaaclab matches the robot embodiment, end-effector type, and task category of the Humanoid Everyday real-world dataset, we use the same Unitree G1-29DoF humanoid robot equipped with Dex3 dexterous hands \cite{Unitree} in both the real-world and simulation domains, and select similar tabletop manipulation tasks. The take\_out\_the\_lid\_of\_the\_spray\_bottle\_g1 scenario in Humanoid Everyday is used to construct $\mathcal{D}_{\mathrm{real}}$. In the unitree\_sim\_isaaclab environment, the Isaac-PickPlace-Cylinder-G129-Dex3-Joint task is used to construct the interactive simulator $\mathcal{M}_{\mathrm{sim}}$ and collect $\mathcal{D}_{\mathrm{sim}}$. 

In Humanoid Everyday, $\mathbf{q}^{\mathrm{xr}}$ includes the 9-dimensional operator head orientation, the 16-dimensional left and right wrist poses, and the 7-dimensional left and right Dex3 dexterous-hand actions, i.e., total dimensions $J_{\mathrm{xr}}=55$. Since the type of each sampling-rate dimension is float64, $d_j=64$ bits. The execution-side $\mathbf{q}^{\mathrm{rb}}$ is constructed from the G1-Dex3 joint states, including 7 joints for each arm, 7 joints for each hand, 6 joints for each leg, and 3 waist joints. Hence, the robot-side has totally $J_{\mathrm{rb}}=43$ joints. Considering that tabletop manipulation task is mainly sensitive to arm and dexterous-hand motions, we set $w_j$ to 1.0 for arm joints, 1.5 for hand joints, 0.5 for leg joints, 0.75 for waist joints in $\mathbf{W}$. Since all trajectories in Humanoid Everyday are recorded at 30 Hz, the motion data collected by Apple Vision Pro are also recorded at 30 Hz. In practice, Apple Vision Pro\cite{apple_vision_pro} can support a 90 Hz sampling rate, while Meta Quest 3\cite{meta_quest3} also provides a 120 Hz high-rate mode. Therefore, we adopt an emulation method of sampling-rate levels. Specifically, four normalized raw-rate levels, i.e., $\{0,10,20,30\}$ Hz, are first constructed from the original 30 Hz motion records through temporal downsampling. These rate levels are then mapped to the target XR communication-rate set $\mathcal{F}_j=\{0,30,60,120\}$ Hz. Under this mapping, the resulting sampled data are further processed through the sampling, transmission, interpolation, and reconstruction pipelines defined in Section \ref{sec:system_model}. This design enables the public 30 Hz dataset to preserve the reconstruction-error differences among different sampling-rate levels, while matching the transmission load of  high-rate XR teleoperation. 
\subsubsection{Algorithm design details} 
$L=5$, $\lambda_\mathrm{MMD}^0=1$, $\lambda_\mathrm{reg}=10^{-5}$, $\alpha_{\ell}=0.2$, $\sigma_{\ell}=2^{\ell-3}\sigma_0$, where $\sigma_0$ is the pairwise distance median after merging simulator and real latent samples in Section \ref{Encoder Warm-up}. $\gamma=0.99$, $\epsilon=0.2$, $\lambda_\mathrm{GAE}=0.95$ in Section \ref{sec:stage1}. $\epsilon_w=10^{-8}$ in Section \ref{sec:IV_C}. Iterations $K_0=500$, $K_1=25$, $K_2=40$, $K_3=10$, trust-region $\beta_0=1$, $\mu=0.95$. 

\subsection{Performance Evaluation}

\begin{figure}
\centering
\begin{subfigure}{1\linewidth}
    \centering
    \includegraphics[height=5cm, keepaspectratio]{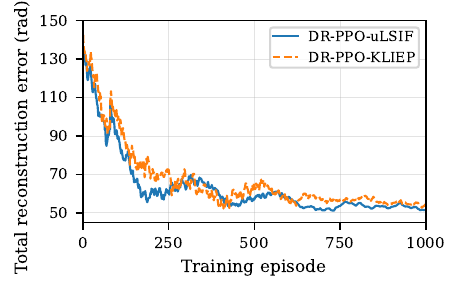}
    \caption{Total reconstruction error in each episode.}
    \label{training_error}
\end{subfigure}
\\  [2ex]
\begin{subfigure}{1\linewidth}
    \centering
    \includegraphics[height=5cm, keepaspectratio]{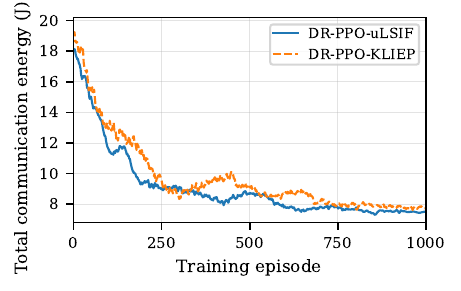}
    \caption{Total effective communication energy in each episode.}
    \label{training_energy}
\end{subfigure}
\caption{Training convergence.}
\label{training_convergence}
\end{figure}
The training convergence is shown in Fig.~\ref{training_convergence}. The total reconstruction error and total effective communication energy are accumulated over $T$ slots within each episode. DR-PPO-uLSIF and DR-PPO-KLIEP denote two weighted PPO algorithms using different density-ratio estimation methods. Both methods eventually reach similar performance. In comparison, DR-PPO-uLSIF converges faster and achieves slightly better performance, whereas DR-PPO-KLIEP exhibits larger fluctuations during training.  The faster and smoother convergence of DR-PPO-uLSIF is mainly attributed to the more stable density-ratio estimates produced by least-squares fitting, while KLIEP is more sensitive to sample variability.

\begin{figure}
\centering
\begin{subfigure}
{1\linewidth}
    \centering
    \includegraphics[height=4.8cm, keepaspectratio]{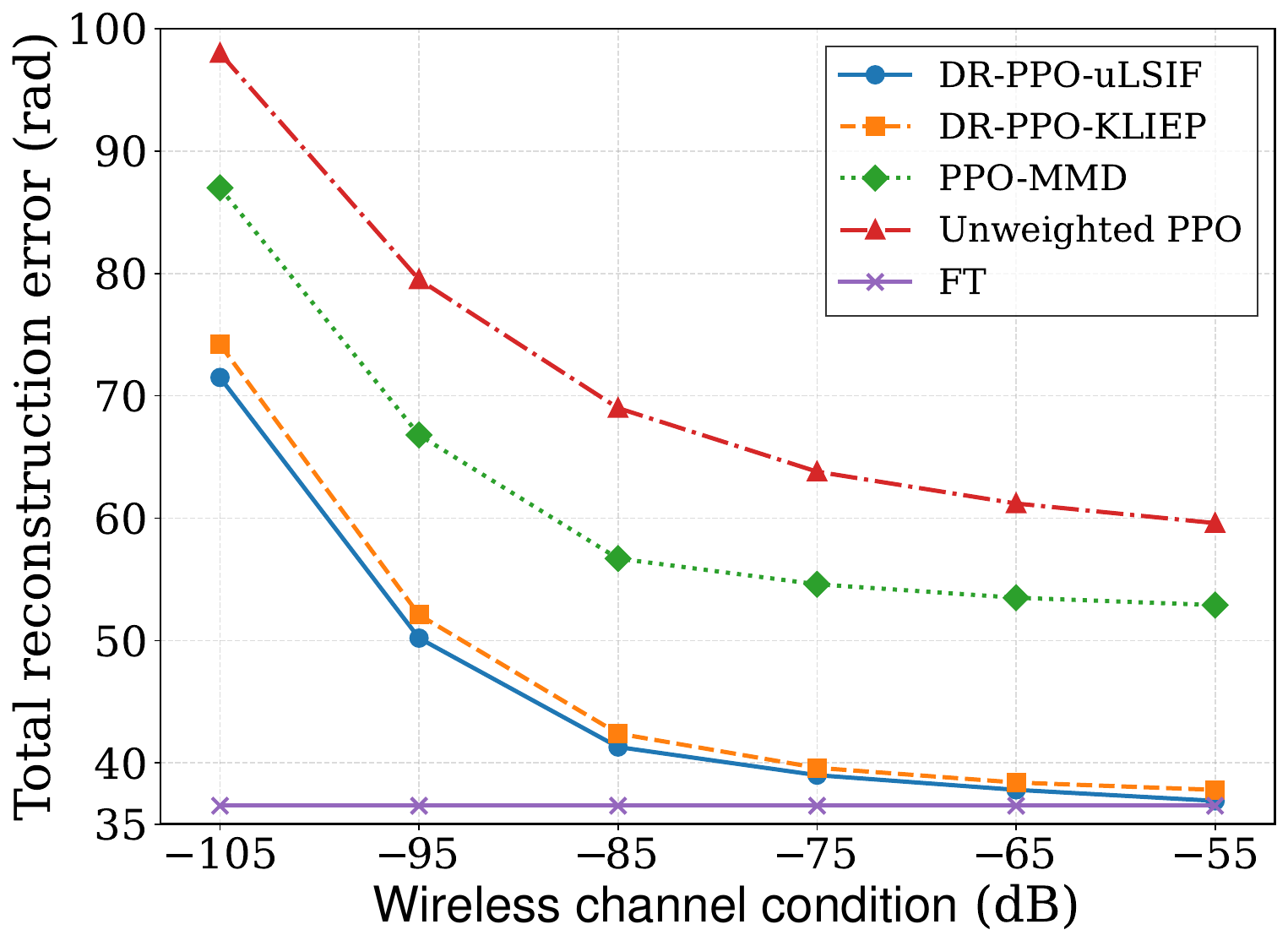}
    \caption{Total reconstruction error of different algorithms.}
    \label{performance_error}
\end{subfigure}
\\  [2ex]
\begin{subfigure}{1\linewidth}
    \centering
    \includegraphics[height=5cm, keepaspectratio]{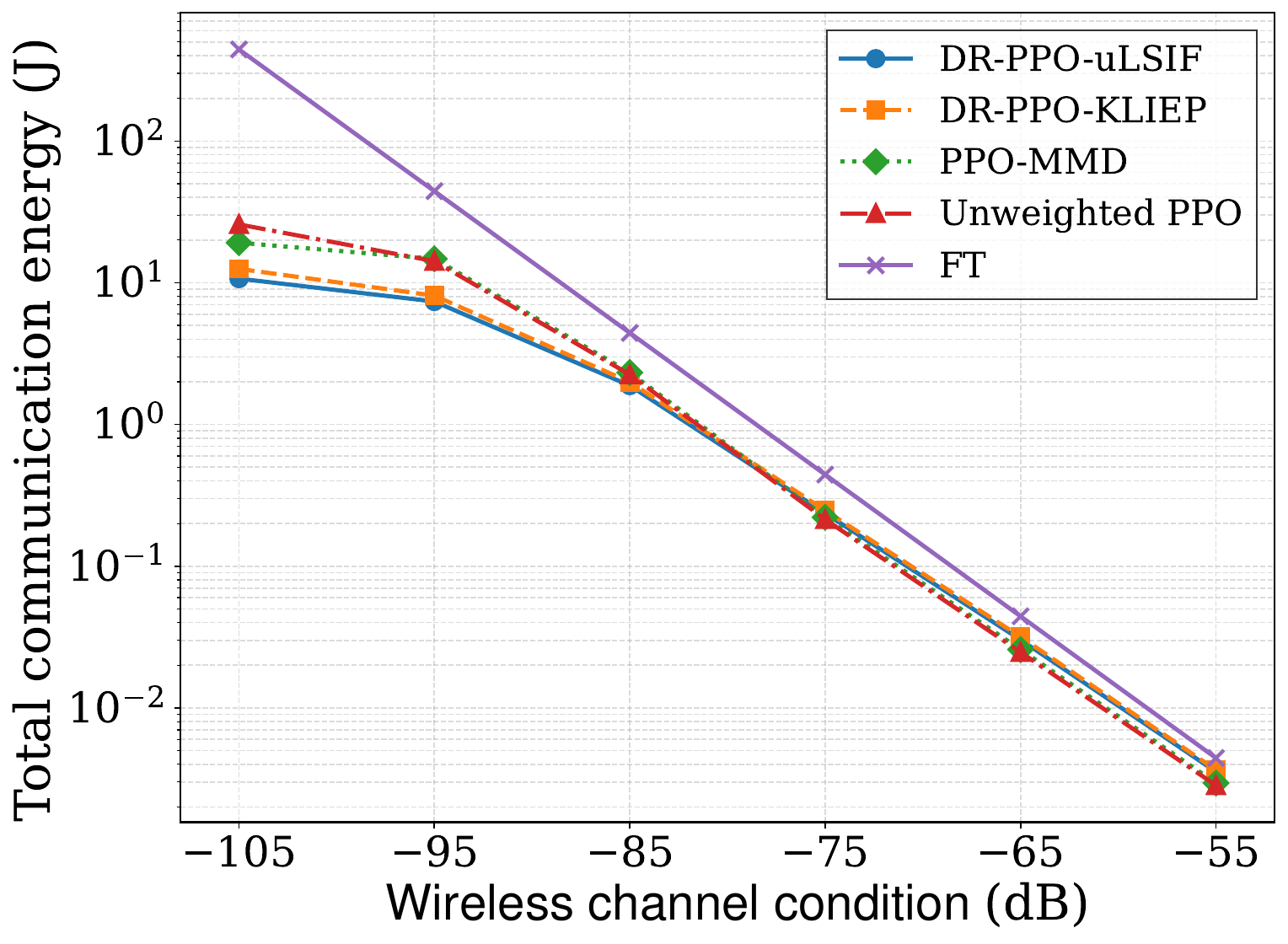}
    \caption{Total effective communication energy of different algorithms.}
    \label{performance_energy}
\end{subfigure}
\caption{Performance comparison under different channel gains.}
\label{Performance}
\end{figure}
Three baselines are included in Fig. \ref{Performance} for ablation. \textbf{PPO-MMD} performs PPO training with MMD-based latent domain but without density-ratio reweighting. \textbf{Unweighted PPO} removes importance weighting and is trained without real-world data. \textbf{FT} directly adopts full-rate transmission without learning policy. The wireless channel condition in dB is given by $10\lg g(t)$. According to Fig.~\ref{performance_error}, the proposed DR-PPO-uLSIF and DR-PPO-KLIEP methods achieve significantly lower reconstruction errors than other policy methods. Unweighted PPO performs the worst due to sim-to-real distribution mismatch. Although FT attains the lowest reconstruction error under severely degraded channel conditions, its energy cost becomes prohibitive. As shown in Fig.~\ref{performance_energy}, when $10\lg g(t)=-105~\mathrm{dB}$, the energy consumption of FT is tens of times higher than that of the policy-learning methods. Moreover, Fig.~\ref{performance_energy} shows that DR-PPO-uLSIF and DR-PPO-KLIEP reduce the effective communication energy to approximately half of that required by Unweighted PPO and PPO-MMD when the channel deteriorates to $-105~\mathrm{dB}$. 
As the channel improves, the proposed methods adaptively exploit the more favorable transmission condition by allocating higher sampling rates, thereby achieving a better tradeoff between reconstruction accuracy and long-term communication energy efficiency.

\begin{figure}
\centering
\begin{subfigure}{1\linewidth}
    \centering
    \includegraphics[height=6.2cm, keepaspectratio]{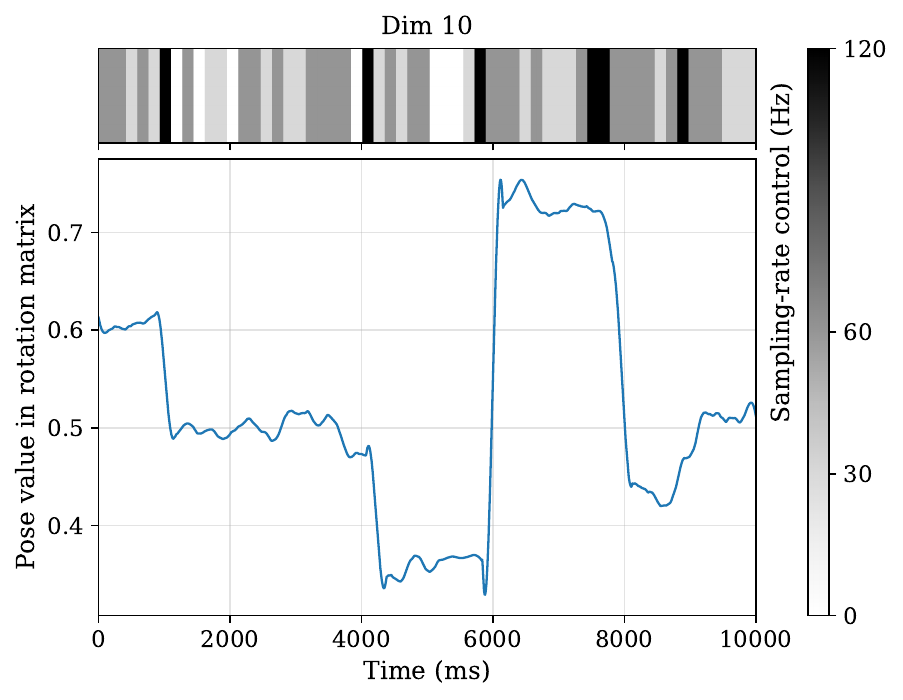}
    \caption{Trajectory and sampling-rate decisions of Dimension 10.}
    \label{dim10}
\end{subfigure}
\\  [2ex]
\begin{subfigure}{1\linewidth}
    \centering
    \includegraphics[height=6.2cm, keepaspectratio]{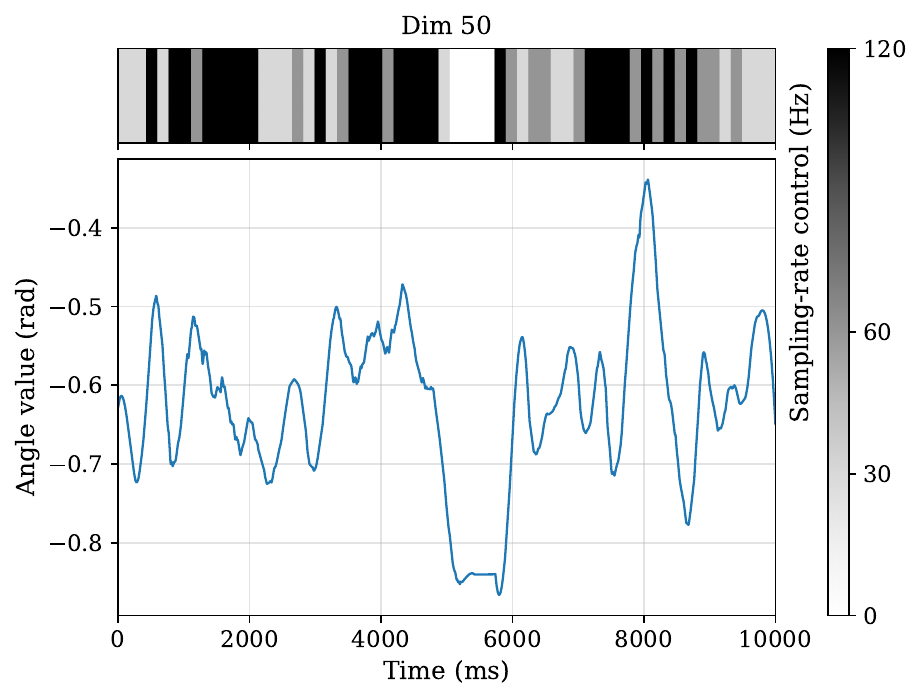}
    \caption{Trajectory and sampling-rate decisions of Dimension 50.}
    \label{dim50}
\end{subfigure}
\caption{Dimension-wise sampling-rate decisions generated by DR-PPO-uLSIF.}
\label{Trajectories}
\end{figure}

Fig.~\ref{Trajectories} present the dimension-wise sampling-rate decisions obtained by the DR-PPO-uLSIF under the channel condition $10\lg g(t)=-80~\mathrm{dB}$. 
The two figures correspond to the same $10000~\mathrm{ms}$ time segment. 
The associated motion trajectories are extracted from take\_out\_the\_lid\_of\_the\_spray\_bottle\_g1, which is held out for validation. By comparing two subfigures, it can be observed that the sampling rates in Fig.~\ref{dim50} are significantly higher than those in Fig.~\ref{dim10}. 
This is because the 10-th dimension, $q^{\mathrm{xr}}_{10}(\nu)$ in Fig.~\ref{dim10}, is the pose value of the left-wrist in rotation matrix, whereas the 50-th dimension, $q^{\mathrm{xr}}_{50}(\nu)$ in Fig.~\ref{dim50} is an angle value of the right Dex3 dexterous hand. 
Compared with the wrist component, the dexterous-hand component contains higher frequency motion, and therefore requires a higher sampling rate to preserve reconstruction accuracy. In Fig.~\ref{dim10}, it can be observed that the sampling rate increases when the pose value exhibits larger temporal variations. 
These results demonstrate that the DR-PPO-uLSIF algorithm can learn motion-adaptive sampling-rate decisions and effectively respond to time-varying XR motion dynamics.

\section{Conclusion}\label{conclusion}

This paper investigated sim-to-real communication-rate optimization for humanoid-robot wireless XR teleoperation. We developed a unified framework integrating motion sampling, wireless transmission, and simulator-physical execution to address the coupled challenges of energy efficiency and sim-to-real mismatch. By establishing a PAC-Bayes generalization characterization based on latent density-ratio estimation and representation bias, we provided theoretical insights into robust policy transfer under distribution shift. Building on this analysis, we proposed a domain-adaptive RL framework with latent density-ratio corrected PPO and trust-region regularization. Experimental results demonstrated superior energy–reconstruction tradeoffs and enhanced robustness under diverse wireless environments and dynamic humanoid motion trajectories. These findings highlight the importance of integrating theoretical generalization analysis with domain-adaptive communication intelligence for future wireless embodied AI systems.

\section*{Appendix: Density Ratio Estimation in\\\ \ Latent Space}
\label{uLSIF and KLIEP}
This appendix summarizes the density-ratio estimation procedures used in Stage~1. The goal is to estimate the real-to-simulator density ratio in the latent space induced by the frozen encoder $\Omega_{\theta_0}$. Let $\{\tilde{x}_u^{\mathrm{sim}}\}_{u=1}^{n_{\iota,s}}$ and 
$\{\tilde{x}_v^{\mathrm{real}}\}_{v=1}^{n_{\iota,r}}$
denote the simulator and real-domain samples used for density-ratio estimation, obtained from $\mathcal S_\iota$ and $\mathcal X_{\mathrm{real}}$, respectively. Their latent
representations under $\Omega_{\theta_0}$ are
\begin{equation}
\tilde{z}_u^{\mathrm{sim}}
=
\Omega_{\theta_0}(\tilde{x}_u^{\mathrm{sim}}), 
\ 
\tilde{z}_v^{\mathrm{real}}
=
\Omega_{\theta_0}(\tilde{x}_v^{\mathrm{real}}).
\end{equation}

For $\dagger\in\{\mathrm{sim},\mathrm{real}\}$, the latent distribution induced by $z=\Omega_{\theta_0}(x)$ with $x\sim P_\dagger$ is defined as
\begin{equation}
\label{eq:latent_pushforward_distribution}
P_{\dagger}^{\theta_0}
=
P_{\dagger}\circ\Omega_{\theta_0}^{-1}.
\end{equation}

Since $P_{\mathrm{real}}\ll P_{\mathrm{sim}}$, the pushforward distributions also satisfy
$P_{\mathrm{real}}^{\theta_0}\ll P_{\mathrm{sim}}^{\theta_0}.$ The target real-to-simulator density ratio in the latent space is defined as
\begin{equation}
\label{eq:latent_target_ratio_appendix}
\iota_{\theta_0}^{\star}(z)
:=
\frac{\mathrm{d}P_{\mathrm{real}}^{\theta_0}}
     {\mathrm{d}P_{\mathrm{sim}}^{\theta_0}}(z).
\end{equation}

To approximate the target ratio in \eqref{eq:latent_target_ratio_appendix}, we use a kernel basis model with $M$ basis functions. Let 
$\xi=[\xi_1,\ldots,\xi_M]^\top$ denote the coefficient vector, where $\xi_m$ is the coefficient of the $m$-th basis function. The parametric density-ratio model is
\begin{equation}
\label{eq:kernel_ratio_expansion}
\iota_{\xi}(z)
=
\sum_{m=1}^{M}\xi_m\phi_m(z),
\ 
\phi_m(z)=\varpi(z,c_m),
\end{equation}
where $\{c_m\}_{m=1}^{M}$ are kernel centers and $\varpi(\cdot,\cdot)$ is the kernel function. The feature vector is defined as
\begin{equation}
\label{eq:kernel_feature_vector}
\boldsymbol{\phi}(z)
=
[\phi_1(z),\ldots,\phi_M(z)]^\top .
\end{equation}

Once the coefficient vector is estimated by uLSIF or KLIEP, denoted by $\hat{\xi}$, the latent density-ratio estimator used in Section~\ref{sec:stage1} is given by $\hat{\iota}(z)=\iota_{\hat{\xi}}(z)$. 
This estimator is the raw density-ratio estimate. Its clipping and normalization are performed separately in Sections~\ref{III-C}, ~\ref{sec:stage1}, and ~\ref{sec:IV_C}.
\subsection*{A.1 uLSIF: Unconstrained Least-Squares Importance Fitting}

uLSIF estimates $\iota_{\theta_0}^{\star}(z)$ by solving the regularized least-squares objective
\begin{equation}
\min_{\mathbf{\xi}}\;
\frac{1}{2} \mathbb{E}_{z\sim  P^{\theta_0}_{\mathrm{sim}}}[ \iota_\mathbf{\xi}(z)^2 ]
-
\mathbb{E}_{z\sim P^{\theta_0}_{\mathrm{real}}}[ \iota_\mathbf{\xi}(z) ]
+ 
\frac{\lambda_\iota}{2}\|\mathbf{\xi}\|_2^2,
\end{equation}
where $\lambda_\iota>0$ is the $\ell_2$ regularization weight.
Using samples $\{\tilde{z}_u^{\mathrm{sim}}\}^{n_\mathrm{\iota,s}}_{u=1}$ and $\{\tilde{z}_v^{\mathrm{real}}\}^{n_\mathrm{\iota,r}}_{v=1}$, we minimize
\begin{equation}
L_{\mathrm{uLSIF}}(\mathbf{\xi})\!
=\!
\frac{1}{2n_\mathrm{\iota,s}}\!\sum_{u=1}^{n_\mathrm{\iota,s}}\! \iota_\mathbf{\xi}(\tilde{z}_u^{\mathrm{sim}})^2\!
-\!
\frac{1}{n_\mathrm{\iota,r}}\!\sum_{v=1}^{n_\mathrm{\iota,r}}\! \iota_\mathbf{\xi}(\tilde{z}_v^{\mathrm{real}})\!
+ \!\frac{\lambda_\iota}{2}\|\mathbf{\xi}\|_2^2.
\end{equation}

Define the empirical matrices
\begin{align}
    H_{m,n}\!=\!\frac{1}{n_\mathrm{\iota,s}}\!\sum_{u=1}^{n_\mathrm{\iota,s}}\!\phi_m(\tilde{z}_u^{\mathrm{sim}})\phi_n(\tilde{z}_u^{\mathrm{sim}}),m,n\!\in\!\{1,\ldots,M\},\\
h_m = \frac{1}{n_\mathrm{\iota,r}}\sum_{v=1}^{n_\mathrm{\iota,r}}\phi_m(\tilde{z}_v^{\mathrm{real}}),\ m\in\{1,\ldots,M\}.
\end{align}
Then, we have
\begin{equation}
    L_{\mathrm{uLSIF}}(\xi) 
= \frac{1}{2}\mathbf{\xi}^\top H\mathbf{\xi} - \mathbf{\xi}^\top h + \frac{\lambda_\iota}{2}\|\mathbf{\xi}\|_2^2,
\end{equation}
\begin{equation}
\hat{\xi} = (H + \lambda_{\iota} I)^{-1} h.
\end{equation}
In implementation, we optionally clip $\iota_{\hat{\xi}}(z)$ to be nonnegative.
\subsection*{A.2 KLIEP: Kullback–Leibler Importance Estimation Procedure}
KLIEP estimates $\iota_{\theta_0}^{\star}(z)$ by minimizing KL divergence
between the target latent distribution and the reweighted simulator latent distribution. Up to constants independent of $\xi$,
\begin{align}
\max_{\xi}\;
\mathbb{E}_{z\sim P^{\theta_0}_{\mathrm{real}}}\!\big[\log \iota_{\xi}(z)\big]\\
\text{s.t.}\quad
\mathbb{E}_{z\sim P^{\theta_0}_{\mathrm{sim}}}\!\big[\iota_{\xi}(z)\big]=1,\\\ \xi \succeq 0.
\end{align}

Using samples $\{\tilde{z}_u^{\mathrm{sim}}\}^{n_{\iota,s}}_{u=1}$ and $\{\tilde{z}_v^{\mathrm{real}}\}^{n_{\iota,r}}_{v=1}$, the objective is denoted by
\begin{align}
\label{eq:kliep-emp}
L_{\mathrm{KLIEP}}(\xi)
&=
-\frac{1}{n_\mathrm{\iota,r}} \sum_{v=1}^{n_\mathrm{\iota,r}} 
\log\!\big(\xi^\top \phi(\tilde{z}_v^{\mathrm{real}})\big)
\nonumber\\&+
\varsigma
\left(\!
\frac{1}{n_\mathrm{\iota,s}}\sum_{u=1}^{n_\mathrm{\iota,s}}\xi^\top\phi(\tilde{z}_u^{\mathrm{sim}})\!-\!1
\!\right)\!^2
+ \frac{\lambda_\iota}{2}\|\xi\|_2^2,
\end{align}
where $\varsigma>0$ controls the constraint penalty, and \eqref{eq:kliep-emp} is under the constraint $\xi \succeq 0$.

The gradient of \eqref{eq:kliep-emp} with respect to $\xi_m$ is given by
\begin{align}
\label{eq:kliep-grad}
\frac{\partial L_{\mathrm{KLIEP}}}{\partial \xi_m}
&=
- \frac{1}{n_\mathrm{\iota,r}} \sum_{v=1}^{n_\mathrm{\iota,r}}
\frac{\phi_m(\tilde{z}_v^{\mathrm{real}})}{\xi^\top\phi(\tilde{z}_v^{\mathrm{real}})}
\nonumber\\
&
+ 2\varsigma
\left(\!
\frac{1}{n_\mathrm{\iota,s}}\sum_{u=1}^{n_\mathrm{\iota,s}}\xi^\top\phi(\tilde{z}_u^{\mathrm{sim}})\!-\!1
\!\right)
\left(\!
\frac{1}{n_\mathrm{\iota,s}}\sum_{u=1}^{n_\mathrm{\iota,s}}\phi_m(\tilde{z}_u^{\mathrm{sim}})
\!\right)
\nonumber\\
&
+ \lambda_\iota \xi_m .
\end{align}

In practice, \eqref{eq:kliep-emp} is optimized by a constrained quasi-Newton method (e.g., L-BFGS-B) to enforce $\xi \ge 0$\cite{L-BFGS-B}.

\bibliographystyle{IEEEtran}
\bibliography{reference}

\end{document}